\documentclass[11pt]{article}
\usepackage[utf8]{inputenc}
\usepackage[T1]{fontenc}   
\usepackage[hidelinks]{hyperref}
\hypersetup{
    colorlinks=true,
    linkcolor=blue,
    filecolor=blue,
    citecolor=blue,
    urlcolor=blue,
}
\usepackage{url}                 
\usepackage{microtype}      
\usepackage{mathtools}
\usepackage{xcolor}   
\usepackage{graphicx}    
\usepackage{physics}
\usepackage{amsmath}
\usepackage{amsthm}
\usepackage{multicol}
\usepackage{cleveref}
\usepackage{dirtytalk}
\usepackage{enumitem}

\usepackage[style=ieee-alphabetic]{biblatex}
\usepackage{libertinus}
\usepackage[libertine]{newtxmath}
\usepackage[
    top=2.6cm,
    bottom=2.6cm,
    left=1.9cm,
    right=1.9cm,
]{geometry}

\usepackage[classfont = bold,
            langfont  = caps,
            funcfont  = roman]{complexity}

\newlang{\UNSAT}{Unsat}
\newlang{\TAUT}{Taut}
\newlang{\GapSVP}{GapSVP}
\newlang{\SIVP}{SIVP}
\newlang{\SIS}{SIS}
\newcommand{\gammaGSVP}{\gamma\text{-}\GapSVP}

\newcommand{\LA}{\mathsf{LA}}
\newcommand{\LAQ}{\mathsf{LA}_\mathbb{Q}}

    \newcommand{\LAint}{\mathsf{int}}

\usepackage{bussproofs}

\theoremstyle{definition}
\newtheorem{definition}{Definition}[section]
\newtheorem{assumption}[definition]{Assumption}

\theoremstyle{plain}
\newtheorem{theorem}[definition]{Theorem}
\newtheorem{lemma}[definition]{Lemma}
\newtheorem{corollary}[definition]{Corollary}
\newtheorem{proposition}[definition]{Proposition}

\newtheorem{innercustomprop}{Proposition}
\newenvironment{customprop}[1]
  {\renewcommand\theinnercustomprop{#1}\innercustomprop}
  {\endinnercustomprop}

\newtheorem*{theorem*}{Theorem}
\newtheorem*{corollary*}{Corollary}

\theoremstyle{remark}
\newtheorem{remark}[definition]{Remark}

\newcommand{\LL}{\mathcal{L}}
\newcommand{\cB}{\mathcal{B}}
\newcommand{\bZ}{\mathbb{Z}}
\newcommand{\bQ}{\mathbb{Q}}
\newcommand{\bR}{\mathbb{R}}
\newcommand{\bN}{\mathbb{N}}
\newcommand{\bF}{\mathbb{F}}

\newcommand{\size}[2]{\mathsf{size}_{#1}(#2)}

\newcommand{\email}[1]{\href{mailto:#1}{\textsf{#1}}}

\title{Quantum Automating $\TC^0$-Frege Is LWE-Hard\footnote{last version: \today}}
\date{}

\author{
    Noel Arteche\footnote{Lund University and University of Copenhagen, \email{noel.arteche@cs.lth.se}} \and
    Gaia Carenini\footnote{École Normale Supérieure (ENS-PSL) and University of Cambridge, \email{gaia.carenini@ens.psl.eu}} \and
    Matthew Gray\footnote{University of Oxford, \email{matthew.gray@cs.ox.ac.uk}}
}

\begin{document}
\maketitle
\abstract{
We prove the first hardness results against efficient proof search by quantum algorithms. We show that under Learning with Errors (LWE), the standard lattice-based cryptographic assumption, no quantum algorithm can weakly automate ${\TC^0\text{-Frege}}$. This extends the line of results of \citeauthor{KP98} (\emph{Information and Computation}, 1998), \citeauthor{BPR1997} (\emph{SIAM Journal on Computing}, 2000), and Bonet, Domingo, Gavaldà, Maciel, and Pitassi (\emph{Computational Complexity}, 2004), who showed that Extended Frege, $\TC^0$-Frege and {$\AC^0\text{-Frege}$}, respectively, cannot be weakly automated by classical algorithms if either the RSA cryptosystem or the Diffie-Hellman key exchange protocol are secure. To the best of our knowledge, this is the first interaction between quantum computation and propositional proof search.
}


\section{Introduction}
\label{sec:intro}
Traditionally, propositional proof complexity has been primarily concerned with proving lower bounds for the length of proofs in propositional proof systems, with the ultimate goal of settling whether $\NP = \coNP$ \cite{cookReckhow}. In parallel, a growing line of research has focused on the computational hardness of finding propositional proofs. Efficient proof search is formally captured by the notion of \emph{automatability}, introduced by \citeauthor{BPR1997} \cite{BPR1997}: a propositional proof system $S$ is automatable if there exists an algorithm that given as input a tautology $\varphi$, outputs an $S$-proof of $\varphi$ in time polynomial in the size of the shortest proof. By relating proofs and computation, automatability connects proof complexity to central areas of theoretical computer science such as automated theorem proving, SAT solving and combinatorial optimization \cite{jakob}, learning theory \cite{AleknovichLearning, pichLearning}, and Kolmogorov complexity \cite{krajicekKolmo}.

Except for very weak proof systems like Tree-like Resolution, automatable in quasipolynomial time \cite{BP96}, most natural systems appear impossible to automate under standard hardness assumptions. Existing hardness results can be split into two broad categories. Work from the late 90s and early 00s showed that stronger proof systems are non-automatable under cryptographic assumptions, while more recent work has shown that weaker proof systems are non-automatable under the optimal assumption that $\P \neq \NP$.

The cryptography-based approach was initiated by the seminal work of \citeauthor{KP98} \cite{KP98}, who showed that Extended Frege is not automatable unless factoring can be solved efficiently, although the notion of automatability would only be defined slightly later by \citeauthor{BPR1997} \cite{BPR1997}, who showed that $\TC^0$-Frege is hard to automate unless Blum integers can be factored by polynomial-size circuits. Finally, Bonet, Domingo, Gavaldà, Maciel, and Pitassi \cite{BDGMP04} extended the existing result from $\TC^0$-Frege to $\AC^0$-Frege under the stronger assumption that Blum integers cannot be factored by subexponential-size circuits.

Building on a long line of work \cite{buss1995godel, Iwa97, pudlak03, AB04, ABMP98, AR08, mertz}, the first $\NP$-hardness result was shown in 2019, when \citeauthor{AM20} \cite{AM20} proved that Resolution is not automatable unless $\P = \NP$. This is optimal, as $\P = \NP$ implies the automatability of any proof system. Their proof uses a clever reduction from $\SAT$ that requires showing a specific lower bound for this system. The technique has since been adapted to other weak proof systems such as Regular and Ordered Resolution \cite{B20}, $k$-DNF Resolution \cite{G20}, Cutting Planes \cite{GKSMP20}, Nullstellensatz and Polynomial Calculus \cite{dRGNRS21}, the OBDD proof system \cite{IR22} and, more recently, even $\AC^0$-Frege \cite{P23}.

Though the latter works prove non-automatability under the optimal hardness assumption, their strength is incomparable to the cryptography-based results. The $\NP$-hardness results all rely on proving specific super-polynomial proof complexity lower bounds for each system, meaning this strategy fails for $\AC^0[2]$-Frege and systems above, for which no lower bounds are known. In contrast, the cryptographic hardness results work by ruling out \emph{feasible interpolation} for these systems, a property which allows one to extract computational content from proofs. For a proof system $S$ proving its own soundness (such as $\TC^0$-Frege), feasible interpolation is equivalent to the notion of \emph{weak} automatability introduced by \citeauthor{AB04} \cite{AB04}, the latter meaning that no proof system simulating $S$ is automatable. The question of whether weak systems such as Resolution are weakly automatable remains one of the major open problems in the field. In short, there exists a trade-off between the strength of the hardness assumption involved ($\P \neq \NP$ versus cryptographic) and the generality of the result (automatability versus weak automatability).

Our work is the first new contribution to the non-automatability of strong proof systems\footnote{We use the terms \emph{weak} and \emph{strong} informally throughout the paper. Traditionally, a strong proof system is a system that proves its own soundness, though it is often also intended to be a system for which lower bounds are lacking. For our purposes, \emph{strong} refers to anything simulating $\TC^0$-Frege, for which both of the previous conditions apply.} in more than two decades. The early results \cite{KP98, BPR1997, BDGMP04} relied on the assumption that factoring is hard, which does not hold for quantum models of computation due to Shor's breakthrough algorithm \cite{shor}. This raises the question of whether a quantum machine could carry out proof search efficiently for some strong proof system. Grover's search algorithm \cite{grover} already provides a quadratic speed-up over brute-force proof search for any system. While this is not enough to achieve automatability, the possibility of more powerful algorithms motivates the interest in new conditional hardness results. The $\NP$-hardness results outlined above imply that $\NP \not\subseteq \BQP$ suffices to rule out automatability for many weak systems, but for stronger systems no widely believed assumption had yet been proven sufficient.

In this work, we formally define the notion of \emph{quantum automatability} and show the first hardness results. We prove that $\TC^0$-Frege is not quantum automatable unless lattice-based cryptography can be broken by polynomial-size quantum circuits. Our results follow from the relationship between automatability and feasible interpolation suitably generalized to the quantum setting. This means that we also rule out quantum feasible interpolation and weak quantum automatability under the same cryptographic assumptions.

\subsection*{Contributions}
Our main contribution is proving the hardness of quantum automatability under the assumption that lattice-based cryptography is secure against quantum computers.

In 1996, Ajtai \cite{ajtai1996generating} gave the first worst-case to average-case reductions for lattice problems. In 1997, in joint work with Dwork \cite{ajtai1997public}, the worst-case hardness of such lattice problems was used to design public-key cryptosystems. Building on similar principles, the Learning with Errors (LWE) assumption of Regev \cite{regev2009lattices} has become the standard post-quantum cryptographic assumption. The LWE assumption is simple to state, surprisingly versatile, and does not seem susceptible to the period-finding technique crucial to Shor's algorithm.

In this work we show that any quantum algorithm that automates $\TC^0$-Frege can be used to break LWE.

\begin{theorem*}[Main theorem, informal]
    If there exists a polynomial-time quantum algorithm that weakly automates $\TC^0$-Frege, then LWE can be broken in polynomial time by a quantum machine.
\end{theorem*}

We then exploit the simulation of $\TC^0$-Frege by $\AC^0$-Frege proofs of subexponential size to extend the result to $\AC^0$-Frege under a slightly stronger assumption, in the style of \citeauthor{BDGMP04} \cite{BDGMP04}.

\begin{corollary*}
    If there exists a polynomial-time quantum algorithm that weakly automates $\AC^0$-Frege, then LWE can be broken in subexponential time by a quantum machine.
\end{corollary*}

In order to properly state and prove these results, we first formally define the notion of quantum automatability for quantum Turing machines. Note that a quantum algorithm might provide a wrong answer with small probability, so we need to be careful in choosing the right definitions. We show that our definition is equivalent to a similar one over uniform quantum circuits, and we verify that it is robust by reproving Impagliazzo's observation that weak automatability implies feasible interpolation, suitably translated to the quantum setting.

\subsection*{Techniques}
The overall structure of the proofs follows the strategy of the previous non-automatability results of \citeauthor{KP98} \cite{KP98} and \citeauthor{BPR1997} \cite{BPR1997}, but the technical details are quite different due to certain complications arising from lattice-based cryptography. We outline below the main hurdles and the techniques used to overcome them.

\paragraph{Quantum feasible interpolation.}
Our result follows from conditionally ruling out feasible interpolation by quantum circuits. As observed by Impagliazzo, weak automatability implies feasible interpolation. We use this observation contrapositively. Suppose that a proof system can prove the injectivity of a candidate one-way function. In the presence of feasible interpolation, we are guaranteed that there exists small circuits capable of inverting the one-way function one bit at a time. If one believes in the security of the cryptographic object, one must conclude that the proof system does not admit feasible interpolation, and in turn that it is not weakly automatable.

For this strategy to work, the candidate one-way function should fulfill two important conditions. First, its definition must be simple enough that the proof system can easily reason about it. For example, RSA requires modular exponentiation to be defined, which is conjectured not to be computable in $\TC^0$. This forced \citeauthor{BPR1997} to use instead the Diffie-Hellman protocol. Second, the candidate one-way function must be injective. The rather technical reason for injectivity is that feasible interpolation allows one to carry out the inversion bit by bit, which does not guarantee retrieving a correct preimage if there are multiple ones.

A few injective one-way functions based on lattice geometry have been proposed throughout the literature, e.g., see \cite{peikert2008lossy, goyal2020perfect, micciancio2012trapdoors}. However, we consider instead a simple scheme for worst-case lattice-based functions that closely resembles the one described by \citeauthor{Micciancio2011} \cite{Micciancio2011}. Such a scheme has the advantage that its injectivity can be easily verified, and that its worst case one way-ness is guaranteed by the assumed hardness of Learning with Errors, which we will now discuss.

\paragraph{Learning with Errors and certificates of injectivity.} 
We base our construction directly on the Learning with Errors assumption. The assumption is simple to define: roughly speaking, a vector $x$ should be hard to recover after being multiplied by some public matrix $A$, and summed with some Gaussian noise, $Ax + \varepsilon$. While the most naive functions based on LWE are not necessarily injective, we can bound the magnitude of the error vectors to construct a family of functions where almost all of the functions are injective. For most matrices $A$, the corresponding function $f_A$ in this family is worst-case one-way assuming the hardness of LWE \cite{Micciancio2011}.

However, the functions being injective and worst-case one-way is not sufficient, because their injectivity needs to be provable inside $\TC^0$-Frege. Unlike with the Diffie-Hellman construction, where a single proof showed the injectivity of the protocol for all generators, here each injective $f_A$ may require a tailored proof of injectivity. Fortunately, most of these $f_A$ can have their injectivity certified by a left-inverse of $A$ together with a short basis for the dual lattice of the $q$-ary lattice spanned by $A$. These short bases not only certify injectivity, but can also be used as trapdoors to invert the function \cite{peikert2016decade}. Though we do not exploit this directly, one may think of the automating algorithm as extracting such trapdoors from proofs. Instead, we use these certificates to prove the injectivity of most $f_A$ inside $\TC^0$-Frege.

With these properties, we can show that feasible interpolation can be used to invert almost all $f_A$, which is sufficient to break LWE and its associated worst-case lattice problems.

\paragraph{Formal theories for linear algebra.}
The most technical component of the previous work on $\TC^0$-Frege and $\AC^0$-Frege was the formalization of many basic properties of arithmetic directly inside the propositional proof systems, which can be quite cumbersome. While we could borrow a large part of the existing formalization of \citeauthor{BPR1997} \cite{BPR1997}, putting it together to carry out arguments about lattice geometry would still be quite convoluted.

Instead, we follow the approach of \citeauthor{KP98}, who showed the injectivity of RSA in Extended Frege by reasoning in Buss's theory $\mathsf{S}^1_2$ of bounded arithmetic. Universal theorems of this first-order theory translate into propositional tautologies with succinct proofs in Extended Frege. For $\TC^0$-Frege and its sequent calculus formalism $\textsf{PTK}$, the corresponding first-order theory of bounded arithmetic is the two-sorted theory $\mathsf{VTC}^0$ introduced by \citeauthor{cookBOOK} \cite{cookBOOK}. This theory is quite expressive and can reason even about analytic functions, as shown by \citeauthor{jerabekVTC0} \cite{jerabekVTC0}. However, since we are mostly interested in statements of matrix algebra, we use the more convenient formal theory $\LA$ for linear algebra introduced by \citeauthor{cooksoltys} \cite{cooksoltys}.

The theory $\LA$ is quantifier-free and operates directly with matrices. It is strong enough to prove their ring properties, but weak enough to allow all theorems in $\LA$ to translate into propositional tautologies with short $\TC^0$-Frege proofs. In order to handle all the concepts required in our arguments, we work over a conservative extension of $\LA$ over the rationals which we show still propositionally translates into $\TC^0$-Frege.

\subsection*{Open problems}
To the best of our knowledge, this is the first interaction between quantum computation and propositional proof search, and we believe further exploration of connections between the two fields is worthwhile. We outline below three open lines of research, ranging from the interaction between quantum computation and proof complexity to a classical problem in the theory of automatability.

\paragraph{Positive results?}
While hardness of proof search in most natural proof systems is now conditionally ruled out under different assumptions, there exists a handful of systems for which no non-automatability results  are known. This is the case for the $\text{Res}(\oplus)$, $\text{Res}(\log)$, Sherali-Adams and Sum-of-Squares proof systems. Could quantum algorithms automate any of these systems efficiently?

Even for proof systems where worst-case hardness is known, could quantum algorithms provide a significant speed-up over brute-force search? Clearly, Grover's algorithm already achieves a quadratic speed-up, but could this be pushed further in some cases?

\paragraph{Quantum proof complexity.}
Hardness results in automatability involve three key elements: the proof system, the hardness assumption and the model of computation for the automating algorithm. In this work we shifted the latter two to the quantum setting, by choosing a post-quantum cryptographic assumption and a quantum model of computation, but the proof systems considered remain classical.

What would it mean to have an inherently quantum proof system? In the same way that Extended Frege can be seen as $\P/\poly$-Frege, could we define a proof system where lines are quantum circuits? This could open the door to a quantum analogue of the Cook-Reckhow program, where showing lower bounds on quantum proof systems would be related to the question of whether $\QMA = \co\QMA$. We note that an analogous approach exists in the field of parameterized complexity, starting with the work of \citeauthor{dantchev2011parameterized} \cite{dantchev2011parameterized}, who defined parameterized proof complexity as a program to gain evidence on the $\W$-hierarchy being different from $\FPT$. As an intermediate step, it would make sense to consider the case of randomized proof systems and the relationship between $\MA$ and $\coMA$, though this has proven to be challenging so far.

We remark that while \citeauthor{Pud09} \cite{Pud09} already defined the notion of quantum derivation rules for propositional proof systems and defined the quantum Frege proof system, his approach is orthogonal to ours, in that those systems are still designed to derive propositional tautologies. In fact, he showed that classical Frege systems simulate quantum Frege systems, though classical Frege proofs cannot be extracted from quantum proofs by a classical algorithm unless factoring is in $\FP$.

\paragraph{Towards generic hardness assumptions.}
Like the original works on weak automatability, our proof requires \emph{concrete} cryptographic assumptions. That is, we assume that some specific candidate one-way function or cryptographic protocol is secure. The reason is that in order to obtain the upper bounds on which to apply feasible interpolation we need concrete formulas to manipulate inside the different proof systems.

A major open problem in the theory of automatability is to disentangle these results from concrete families of candidate one-way functions. That is, can we prove that $\TC^0$-Frege is not (weakly) automatable under the assumption that, say, one-way functions exist? Even better, can one obtain $\NP$-hardness of automating strong proof systems without the need to prove lower bounds first, in a way different from the strategy of \citeauthor{AM20} \cite{AM20}? This seems to require conceptual breakthroughs. 

\subsection*{Structure of the paper}
The paper is structured as follows. \Cref{sec:prelim} recalls the necessary concepts in proof complexity and lattice-based cryptography needed in the rest of the paper. \Cref{sec:qauto} defines automatability for quantum Turing machines and uniform quantum circuits and proves the equivalence between both models to then reprove Impagliazzo's observation on the relation between automatability and feasible interpolation, now in the quantum setting. \Cref{sec:outline} states and proves the main theorem of the paper. The section first presents a detailed overview of the main argument, while the subsections contain all the necessary technical work.

\section{Preliminaries}
\label{sec:prelim}

We assume basic familiarity with computational complexity theory, propositional logic and quantum circuits. We review the main concepts needed from proof complexity and refer the reader to standard texts like \cite{krajicekBOOK} for further details. We also recall some relevant notions from linear algebra and lattice geometry useful in our arguments.

\subsection{Proof complexity}
\label{subsec:proof-complexity}
Following Cook and Reckhow \cite{cookReckhow}, a \emph{propositional proof system} $S$ for the language $\TAUT$ of propositional tautologies is a polynomial-time surjective function $S : \{0,1\}^* \to \TAUT$. We think of $S$ as a proof checker that takes some proof $\pi\in \{0 ,1\}^*$ and outputs $S(\pi) = \varphi$, the theorem that $\pi$ proves. Soundness follows from the fact that the range is exactly $\TAUT$, and completeness is guaranteed by the fact that $S$ is surjective. One may alternatively define proof systems for refuting propositional contradictions. We move from one setup to the other depending on context.

We denote by $\size{S}{\varphi}$ the \emph{size} of the smallest $S$-proof of $\varphi$ plus the size of $\varphi$. We say that a proof system $S$ is \emph{polynomially bounded} if there exists a constant $c$ such that for every $\varphi \in \TAUT$, $\size{S}{\varphi} \leq |\varphi|^c$. We say that a proof system $S$ \emph{polynomially simulates} a system $Q$ if there exists a constant $d$ such that for every $\varphi \in \TAUT$, $\size{S}{\varphi} \leq \size{Q}{\varphi}^{d}$. For a family $\{ \varphi_n\}_{n\in\mathbb{N}}$ of propositional tautologies, we write $S \vdash \varphi_n$ whenever $\size{S}{\varphi_n} \leq |\varphi_n|^{c}$ for some constant $c$. Finally, a proof system $S$ is said to be \emph{closed under restrictions} if there is a constant $d$ such that whenever $S$ proves a formula $\varphi$ in size $s$, for every partial restriction $\rho$ to the variables in $\varphi$, there exists a proof of the restricted formula $\varphi_{\restriction\rho}$ in size $s^{d}$.

The focus of this work is on a specific class of proof systems known as  \emph{Frege systems}. A Frege system is a finite set of axiom schemas and inference rules that are sound and implicationally complete for the language of propositional tautologies built from the Boolean connectives negation ($\neg$), conjunction ($\land$), and disjunction ($\lor$). A Frege proof is a sequence of formulas where each formula is obtained by either substitution of an axiom schema or by application of an inference rule on previously derived formulas. As long as the set of inference rules is finite, sound and implicationally complete, the specific choice of rules does not effect the size of the proofs up to polynomial factors, as all Frege systems polynomially simulate each other (see, e.g.\ \autocite[Theorem 4.4.13]{krajicekBOOK}).

We can make gradations between Frege systems by restricting the complexity of their proof lines. For a circuit class $\mathcal{C}$, the system $\mathcal{C}$-Frege is any Frege system where lines are restricted to be $\mathcal{C}$-circuits (see \cite{jerabek2005} for a formal definition). In this setup, a standard Frege system amounts to $\NC^1$-Frege. We are mostly interested in the weaker systems $\AC^0$-Frege and $\TC^0$-Frege, where the proof lines are, respectively, circuits of constant depth and unbounded fan-in, and threshold circuits of constant depth and unbounded fan-in. A \emph{threshold circuit} is a Boolean circuit where gates can be the usual $\neg, \vee, \wedge$ as well as the threshold ones $\mathsf{Th}_k(x_1,\dots, x_n)$, where $\mathsf{Th}_k$ is true if at least $k$ of its inputs are true. For a concrete depth $d$, we denote the corresponding systems by $\AC^0_d$-Frege and $\TC^0_d$-Frege.

It is often convenient to consider an alternative formalism of $\TC^0$-Frege in the style of Gentzen's sequent calculus. The \emph{Propositional Threshold Calculus} $\mathsf{PTK}$ \cite[Chapter X.4.1]{cookBOOK} is a version of the propositional sequent calculus where the cuts are restricted to threshold formulas of constant depth. We refer to \cite[Section 2]{BPR1997} for a complete rendering of the derivational rules of $\mathsf{PTK}$.
 
\subsection{Lattice geometry}
\label{subsec:lattice-geo}
We recall some basic definitions from lattice geometry. For a linearly independent set of $n$ vectors $\mathcal{B}=\{b_1,\dots, b_n\} \subseteq\mathbb{R}^m$, which we often treat simply as an $m\times n$ matrix, the \emph{lattice} over $\mathcal{B}$ is defined to be the set of all integer linear combinations of vectors in $\mathcal{B}$,
\begin{equation*}
\mathcal{L}(\mathcal{B}) \coloneqq \{x \in\mathbb{R}^m \mid \text{there is } a\in\mathbb{Z}^n \text{ such that } x= \mathcal{B}a\}\,.
\end{equation*}

When the vectors in $\mathcal{B}$ belong in $\mathbb{Z}_q^m$ for some modulus $q$, we can further define a \emph{modular lattice} over $\mathcal{B}$, denoted $\mathcal{L}_q(\mathcal{B})$, to be the set of all integer linear combinations of the basis modulo $q$,
\begin{equation*}
\mathcal{L}_q(\mathcal{B}) \coloneqq \{x \in \mathbb{Z}_q^m \mid \text{there is } a \in \mathbb{Z}_q^{n}\text{ such that }\mathcal{B} a \equiv x \text{ mod } q\}\,,
\end{equation*}
where the mod function is applied element-wise in the vector.

We define the length of a vector $x$ in $\LL_q(\cB)$ to be the Euclidean norm of the shortest vector in $\bZ^m$ that is congruent to $x$ modulo $q$. Note that these shortest vectors will always fall in the domain ${[-\lfloor q/2 \rfloor, \lfloor {q}/{2} \rfloor]^m}$.

A \emph{$q$-ary lattice} $\Delta_q(\mathcal{B})$ can be thought of as an extension of a modular lattice back to $\bZ^m$ and is the set of all vectors $x\in\mathbb{Z}^m$ congruent to members of the modular lattice,
\[ \Delta_q(\mathcal{B}) \coloneqq \{x \in \mathbb{Z}^m \mid \text{there is } a \in \mathbb{Z}^{n}\text{ such that } \mathcal{B}a \equiv x \text{ mod } q\} \,.  \]

Note that because for all $x \in\{0,q\}^m$,  $x \in \Delta_q(\mathcal{B})$, we have that all $q$-ary lattices have rank $m$.

From the definitions above it is clear that $\LL_q(\cB) \subseteq \Delta_q(\cB)$. Consequently a proof that no vector in $\Delta_q(\cB)$ has length less than $\ell$ also proves that no vector in $\LL(\cB)$ has length less than $\ell$.

Another important concept in lattice geometry is that of a \emph{dual lattice}. Given a lattice $\mathcal{L}(\mathcal{B})$, its dual lattice $\mathcal{L}^*(\mathcal{B})$ is defined to be the set of vectors within the subspace spanned by $\mathcal{B}$ whose inner product with any element in $\mathcal{L}$ is an integer. Formally,
 \begin{equation*}
     \mathcal{L}^*(\mathcal{B}) \coloneqq \{ y \in \mathbb{R}^m \mid \text{there is } z \in \mathbb{R}^n \text{ such that } y=\mathcal{B} z  \text{ and for all }x \in \mathcal{L}(\mathcal{B}), \langle x, y\rangle \in \mathbb{Z}\}\,,
 \end{equation*}
 where $\langle \cdot, \cdot \rangle$ denotes the inner product. The dual lattice is also a lattice, whose basis admits a closed form.
\begin{lemma}
\label{lem:dual-basis}
For a basis $\mathcal{B} \in \bR^{m\times n}$, $\mathcal{L}^*(\mathcal{B}) = \mathcal{L}(\mathcal{B}(\mathcal{B}^\intercal \mathcal{B})^{-1})$.
\end{lemma}
This lemma is standard and can be found, for example, in \cite{Micciancio2011}. Note that, if $\mathcal{B}\in \mathbb{Z}^{m \times n}$, it is easy to show that $ \mathcal{B}(\mathcal{B}^\intercal \mathcal{B})^{-1} \in \mathbb{Q}^{m \times n}$, and, therefore, that any $x \in \mathcal{L}^*(\mathcal{B})$ belongs to $\mathbb{Q}^m$. 

While a modular lattice and the $q$-ary lattice that extends it are closely related, they do have distinct bases. We use the following fact that given a matrix $\mathcal{B} \in\bZ_q^{m\times n}$ such that $\rank(\mathcal{B}) = n$, there exists a closed form for a matrix $\mathcal{B}'$ such that $\LL(\mathcal{B}')= \Delta_q(\mathcal{B})$. We defer the proof to \Cref{app:q-ary}.
\begin{lemma}[Full-rank modular lattices have $q$-ary lattice bases]\label{lem:q-basis}
    Let $\cB \in \bZ_{q}^{m \times n}$ and define $C \in \{0,1\}^{m \times m}$ to be the permutation matrix that swaps the appropriate rows so that the first $n$ rows of $C\cB$ are linearly independent. Let $(C\cB)_1 \in \bZ^{n \times n}$ and $(C\cB)_2 \in \bZ^{m - n \times n}$ be matrices such that \mbox{$C\cB = [(C\cB)_1^\intercal \mid (C\cB)_2^\intercal]^\intercal$}. Then, for $\cB \in \bZ_q^{m \times n}$, if $\rank(\cB) = n$, $\Delta_q(\cB) = \LL(\cB')$, where
    \[ \cB' = 
    C
    \begin{bmatrix}
    I_n & 0\\
    (C\cB)_2(C\cB)_1^{-1} & qI_{m-n} 
    \end{bmatrix}
    C^{-1}
    \,, \]
    and where the inverses $M^{-1}$ are defined over the modular lattice $\bZ_q^m$.
\end{lemma}

Note that we can combine this with \Cref{lem:dual-basis} to get a closed form for $\mathcal{B}'$ such that $\Delta_q^*(\mathcal{B}) = \LL(\mathcal{B}')$.

The \emph{$i$-th successive minimum} of a lattice $\LL$ is $\lambda_i(\mathcal{L}) \coloneqq \inf \{ r \in \mathbb{Z} \mid \dim (\operatorname{span} (\mathcal{L} \cap B(0, r))) \geq i \}$, where $B(0,r)$ is the ball of radius $r$ around the origin. Less formally,  $\lambda_i(\LL)$ is the length of the $i$-th smallest linearly independent vector in the lattice.

There exists an intimate relationship between a lattice and its dual, as captured by Banaszczyk's Transference Theorem.
\begin{theorem}[Transference Theorem \cite{banaszczyk1993new}]
\label{thm:transference}
 For any rank-$n$ lattice $\mathcal{L} \subseteq \bZ^m$, $1\leq\lambda_1(\mathcal{L})\cdot\lambda_n(\mathcal{L}^*)\leq n$.
\end{theorem}

Modular lattices $\LL_q(\mathcal{B})$ are subsets of $\bZ_q^m$, not $\bZ^m$, and therefore the Transference Theorem does not directly apply. However we are able to leverage the fact that 
 $\lambda_1(\Delta_q(\mathcal{B})) = \min(q, \lambda_1(\LL_q(\mathcal{B})))$ to indirectly apply it through the $q$-ary lattice.

We recall useful properties of random lattices. 

\begin{lemma}\label{lem:random-lattices}
  For a randomly sampled matrix $A \in \bZ_q^{m \times n}$, we have that
  \begin{itemize}
      \item[(i)] $\Pr_A[\rank(A) = n] \geq 1 - {n}/{q^{m-n+1}}$;
      \item[(ii)] $\Pr_A[\lambda_1(\LL_q(A)) < r \mid \rank(A) = n] \leq {(2r+1)^m}/{q^{m-2n-1}}$.
  \end{itemize}
\end{lemma}

These properties are folklore. For the sake of completeness, we provide proofs in \Cref{app:counting}.

\subsection{Learning with Errors (\textsc{LWE})}
\label{subsec:lwe}

Learning with Errors (LWE) is a central problem of learning theory, introduced by Regev \cite{regev2009lattices}. 

For the sake of completeness, we introduce first the definition of discrete Gaussians, although the precise notion is not relevant for our proof. The essential point is that with high probability the error will be at most a few times the standard deviation times $\sqrt{m}$ \cite{regev2009lattices, Micciancio2011}. We follow here the definition of Peikert \autocite[Section 2.3]{peikert2016decade}.

\begin{definition}
    The \emph{discrete Gaussian} with standard deviation (or width) $w$ is defined to be the probability distribution over $\bZ^m$ where the probability of vector $x$ is proportional to $e^{-\pi\abs{\abs{x}}^2/w^2}$.
\end{definition}

\begin{assumption}[The Learning with Errors (LWE) assumption \cite{regev2009lattices, peikert2016decade}]
\label{ass:lwe}
    Let $n \in \bN$, $m = n^{O(1)}$, $q \leq 2^{n^{O(1)}}$, let $s \sim \bZ_q^n$ be a secret vector, $A \sim \bZ_q^{m \times n}$, and $\varepsilon \in \bZ_q^m$ a sample from the discrete Gaussian with standard deviation $c = \alpha q$ with $\alpha = o(1)$ and $\alpha \in [0,1]$. The \emph{Learning with Errors assumption} states that there is no quantum inverter\footnote{In some instances there may also be some $s'$ and small enough
    $\varepsilon'$ such that  $As'+ \varepsilon' = As+\varepsilon$, in which case $s'$ would 
    also be a valid inversion of $(A,As+\varepsilon)$. However, as we discuss later, there are exponentially few matrices $A$ for which any such $s'$ will exist together with some 
    $\varepsilon'$ small enough. Thus, defining the problem in terms of unique inversion (as done by Regev \cite{regev2009lattices} and Piekert \cite{peikert2016decade}) is asymptotically equivalent to a more complex definition accounting for non-unique inversion.
    } $M$ running in time $n^{O(1)}$ such that $M(A, As+ \varepsilon)$ outputs $s$ 
    with noticeable probability
    over the choice of $s$, $A$, $\varepsilon$, and the internal randomness of $M$.
\end{assumption}

Note that it suffices to succeed in the above game to output $As$ with noticeable probability, as you can recover $s$ from $As$ via Gaussian elimination.

The security of this assumption relies on the existence of worst-case to average-case reductions to fundamental lattice problems conjectured to be hard. In particular, as shown by \citeauthor{regev2009lattices} \cite{regev2009lattices}, breaking LWE implies solving the $\gammaGSVP$ problem for an approximation factor $\gamma = n^2$. Here, $\gammaGSVP$ refers to the \emph{$\gamma$-Approximate Shortest Vector Problem}: given a lattice basis $\mathcal{B}\in\mathbb{Q}^{m\times n}$ and a distance threshold $r>0$, decide whether $\lambda_1(\mathcal{L(\cB}))\leq r$, or
$\lambda_1 (\mathcal{L(\cB)})>\gamma r$, when one of those cases is promised to hold.

The belief that $\gammaGSVP$ is intractable is backed by the fact that the problem is $\NP$-hard under randomized reductions when the approximation factor is constant \cite{ajtai1996generating,peikert2016decade, BP2022}. However, for the range of $\gamma$ in which the reduction to LWE works, $\NP$-hardness is not known. Obtaining $\NP$-hardness for polynomial approximation factors would imply the breakthrough consequence of basing cryptography on worst-case hardness assumptions. In turn, this would turn our non-automatability results into $\NP$-hardness results. As appealing as this might be, it is unlikely. For $\gamma \geq \sqrt{n}$, the problem $\gammaGSVP$ is known to be in $\NP \cap \coNP$ \cite{NPcoNP} and thus cannot be $\NP$-hard unless $\PH$ collapses.

\subsection{The formal theory $\LA$}
\label{subsec:LA-prelim}
The theory $\LA$ is a quantifier-free theory introduced by \citeauthor{cooksoltys} \cite{cooksoltys} whose main objects are matrices. This is not technically speaking a first-order theory of bounded arithmetic like those used by Krajíček and Pudlák \cite{KP98}, but like them it admits a propositional translation into Frege systems.

The system $\LA$ operates over three sorts: \emph{indices} (intended to be natural numbers), \emph{field elements} (over some abstract field $\bF$), and \emph{matrices} (with entries over $\bF$). Variables for these three sorts are usually denoted $i,j,k, \dots$ for indices, $a,b,c, \dots$ for field elements, and $A,B,C,\dots,$ for matrices. We sometimes use lower-case letters $v, w, \dots$ for vectors, which are seen as a special case of matrices. 

The language of $\LA$ consists of the following constant, predicate and function symbols, over the three different sorts:

\begin{itemize}
\item {Index sort}: \indent $0_{\text{index}}, 1_{\text{index}}, +_{\text{index}}, \cdot_{\text{index}}, -_{\text{index}}, \mathsf{div}, \mathsf{rem}, \mathsf{cond}_{\text{index}}, \leq_{\text{index}}, =_{\text{index}}$
    \item{Field sort}: $0_{\text{field}}, 1_{\text{field}}, +_{\text{field}}, \cdot_{\text{field}}, -_{\text{field}}, ^{-1}, \mathsf{r}, \mathsf{c}, \mathsf{e}, \Sigma, \mathsf{cond}_{\text{field}},  =_{\text{field}}$
    \item{Matrix sort}: $=_{\text{matrix}}$
\end{itemize}

The meaning of the symbols is the standard one, except for $-_{\text{index}}$ that denotes the cutoff subtraction ($i-j=0$ if $i< j$) and for $a^{-1}$, denoting the inverse of a field element $a$, with $0^{-1}=0$. For operations over matrices, $\mathsf{r}(A)$ and $\mathsf{c}(A)$ are, respectively, the number of rows and columns in $A$, $\mathsf{e}(A,i,j)$ is the field element $A_{i,j}$ (with $\mathsf{e}(A,i,j) = 0$ if either $i=0$, $j=0$, $i>\mathsf{r}(A)$ or $j>\mathsf{c}(A)$) and $\Sigma A$ is the sum of the elements in $A$. The function symbol $\mathsf{cond}(\alpha, t_1, t_2)$ is interpreted to mean that if $\alpha$ holds, then the returned value should be $t_1$, else $t_2$, where $\alpha$ is a formula all of whose atomic subformulas have the form $m\leq n$ or $m=n$, where $m$ and $n$ are of the index sort, and $t_1, t_2$ are terms either both of index sort or both of field sort.

The language of $\LA$ can be enriched with the following defined terms: index maximum ($\mathsf{max}$), matrix sum ($+$, when sizes of the matrices are compatible), scalar product ($\cdot$), matrix transpose ($A^\intercal$), zero ($0$) and identity matrices ($I$), matrix trace ($\tr$), dot product ($\langle \_, \_\rangle$), and matrix product ($\cdot$). See \autocite[Section 2.1]{cooksoltys} for details on the definitions of these terms. In general, whenever it is clear from context, we drop the subscripts indicating the sort and we use standard linear algebra notation for the sake of readability.

The theory then consists of several groups of axioms fixing the meaning of these symbols. These are rather lengthy to state, so we relegate them to \Cref{app:LA-axioms}, where we also include several theorems derived by \citeauthor{cooksoltys} inside $\LA$.

Observe that the theory is field-independent, but whenever we fix the field to be either finite or $\bQ$, $\LA$ has the robust property that every theorem translates into a family of propositional formulas with short $\TC^0$-Frege proofs. This is the main property of $\LA$ that we shall exploit.

\section{Quantum automatability and feasible interpolation}
\label{sec:qauto}

Following Bonet, Pitassi, and Raz \cite{BPR1997}, we say that a propositional proof system $S$ is \emph{automatable in time $t$} if there exists a deterministic Turing machine $A$ that on input a formula $\varphi$ outputs an $S$-proof of $\varphi$, if one exists, in time $t(\mathsf{size}_S(\varphi))$. We now consider the possibility of replacing $A$ by a probabilistic or quantum Turing machine. The main issue in the definition is now that the output of the machine may be erroneous, albeit with small probability. Note, however, that if a machine were to output an incorrect proof, we would be able to easily detect this, since we can verify the proofs in polynomial time. We may thus assume that when yielding an incorrect proof, the machine will restart and find another one. Hence, instead of asking for the error-probability of the machine to be bounded, we ask for the expected running time to be bounded. The following definition captures this idea.

\begin{definition}[Quantum and randomized automatability]
\label{def:q-aut}
Let $S$ be a propositional proof system and let $t : \mathbb{N} \to \mathbb{N}$ be a time-constructible function. We say that $S$ is \emph{quantum} (respectively, \emph{random}) \emph{automatable in time $t$} or simply \emph{quantumatable in time $t$} if there exists a quantum Turing machine (respectively, a randomized Turing machine) that on input a formula $\varphi$ outputs an $S$-proof of $\varphi$, if one exists, in expected time $t(\mathsf{size}_S(\varphi))$.
\end{definition}

In what follows, we assume $t$ to be a polynomial and talk simply about a system being \emph{automatable} or \emph{quantum automatable}, without reference to $t$. Since quantum circuits are often more convenient than quantum Turing machines, we also define automatability in the circuit setting. For this, we use the standard notion of $\P$-uniformity: a circuit family is uniform (or $\P$-uniform) if there exists a polynomial-time Turing machine which on input $1^\ell$ outputs a description of the circuit that solves the problem on inputs of size $\ell$.

\begin{definition}[Circuit automatability]
\label{def:circ-aut}
Let $S$ be a propositional proof system. We say that $S$ is \emph{circuit-automatable} if there exists a constant $c$ and a uniform multi-output circuit family $\{ C_{n, s}\}_{n, s\in \bN}$ of size $(n + s)^c$ such that $C_{n, s}$ takes as input a formula $\varphi$ of size $n$ and outputs an $S$-proof of size $s^c$ if a proof of size $s$ exists, and is allowed to output any string otherwise.
\end{definition}

The generalization to randomized and quantum circuits is now immediate.

\begin{definition}
\label{def:q-cir-aut}
Let $S$ be a propositional proof system. We say $S$ is \emph{quantum circuit-automatable} if there exists a constant $c$ and a uniform multi-output quantum circuit family $\{ C_{n, s}\}_{n, s\in \bN}$ of size $(n+s)^c$ such that $C_{n, s}$ takes as input a formula $\varphi$ of size $n$, and outputs an $S$-proof of size $s^c$ with probability at least $2/3$ if a proof of size $s$ exists, and is allowed to output any string otherwise. We say that $S$ is \emph{random circuit-automatable} if the circuit is classical but also takes as input a sequence $r$ of random bits and, for at least $2/3$s of the choices for $r$, $C_{n,s}(\varphi, r)$ outputs an $S$-proof of size $s^c$ if a proof of size $s$ exists, and is allowed to output any string otherwise.
\end{definition}

The machine-based and circuit-based definitions are equivalent.

\begin{proposition}
\label{prop:equivalences}
Let $S$ be a propositional proof system. The following equivalences hold:
    \begin{enumerate}
        \item[(i)] the system $S$ is automatable if and only if it is circuit-automatable;
        \item[(ii)] the system $S$ is random automatable if and only if it is random circuit-automatable;
        \item[(iii)] the system $S$ is quantum automatable if and only if it is quantum circuit-automatable.
    \end{enumerate}
\end{proposition}

We defer the rather simple proof to \Cref{app:equivalence}.

Even if a proof system is not automatable, one might still hope for an algorithm that finds some proof efficiently, even if it is in a different proof system. We say that a proof system $S$ is \emph{weakly automatable} if there exists another proof system $Q$ and an algorithm $A$ that given a formula $\varphi$, outputs a $Q$-proof of $\varphi$ in time $\size{S}{\varphi}^{O(1)}$. The concept was introduced by \citeauthor{AB04} \cite{AB04}, who further showed that this is equivalent to $S$ being simulated by a system $Q$ that is itself automatable \autocite[Thm.\ 1]{AB04}. Despite the fact that weak automatability has been conditionally ruled out for Resolution under hardness assumptions for certain two-player games \cite{atseriasManeva, huangPitassi, BPT14}, establishing whether weak proof systems---such as Resolution---are weakly automatable under more standard hardness conjectures remains one of the main open problems in the area. It is straightforward to extend the notion of weak automatability to the quantum setting, in the style of \Cref{def:q-aut} and \Cref{def:q-cir-aut}.

Weak automatability is closely related to feasible interpolation. We recall this connection in its classical form and then move to the quantum setting.

\begin{definition}[Feasible interpolation \cite{Kra97, pudlak03}]
\label{def:fi}
We say that a proof system $S$ has the \emph{feasible interpolation property} if there exists a polynomial-time computable function $I$ such that for every tautological split formula $\varphi(x, y, z) = \alpha(x, z) \lor \beta(z, y)$, whenever a proof $\pi$ in $S$ derives $\varphi$ in size $s$, $I(\pi)$ produces an \emph{interpolant circuit} $C_{\varphi}$ of size $s^{O(1)}$ that takes as input an assignment $\rho$ to the $z$-variables and such that
\begin{equation*}
    C_\varphi(\rho) = \begin{cases}
        0 &\text{ only if } \alpha(x, \rho) \text{ is a tautology} \\
        1 &\text{ only if } \beta(\rho, y) \text{ is a tautology} \\
    \end{cases}
\end{equation*}
indicating which side of the conjunction is tautological.
\end{definition}

Bonet, Pitassi, and Raz attribute the following crucial observation relating (weak) automatability and feasible interpolation to Impagliazzo. We refer to it as \emph{Impagliazzo's observation}.

\begin{proposition}[Impagliazzo's observation {\autocite[Thm. 1.1]{BPR1997}}]
    If a proof system is weakly automatable and closed under restrictions, then it admits feasible interpolation.
\end{proposition}

Impagliazzo's observation is useful contrapositively: to rule out (weak) automatability it suffices to rule out feasible interpolation, as done in the previous works \cite{KP98, BPR1997}. We outline this strategy further in \Cref{sec:outline}, where we instantiate it together with our cryptographic assumption.

To use feasible interpolation in our setting, we suitably adapt the definition to the quantum world.

\begin{definition}[Quantum feasible interpolation]
\label{def:qfi}
We say that a proof system $S$ has the \emph{quantum feasible interpolation property} if there exists a polynomial-time computable function $I$ such that, for every tautological split formula $\varphi(x, y, z) = \alpha(x, z) \lor \beta(z, y)$, whenever a proof $\pi$ derives $\varphi$ in $S$ in size $s$, $I(\pi)$ prints the description of a quantum interpolant circuit $C_\varphi$ of size $s^{O(1)}$ as in \Cref{def:fi}. If the circuit is instead randomized, we call this property \emph{random feasible interpolation}.
\end{definition}

Interestingly, feasible interpolation is not affected by moving from classical automatability to randomized automatability. This is essentially folklore, but we reprove it for the sake of completeness.

\begin{proposition}
\label{prop:rand-impa}
If a proof system $S$ is weakly random automatable and closed under restrictions, then it has feasible interpolation by deterministic Boolean circuits.
\end{proposition}

\begin{proof}
    The proof is essentially the same as the original proof in \cite{BPR1997}, except for having to take randomness into account. Suppose $R$ is a probabilistic automating algorithm for $S$. By \Cref{prop:equivalences}.(ii), we can instead think of a family of randomized circuits $\{C_{n, s}\}_{n, s\in \bN}$ that, for some fixed constant $c$, outputs proofs of size $s^c$ when a proof of size $s$ exists. Furthermore, let $d$ be the constant in the exponent that bounds the blow-up in size happening in the closure under restrictions. Given a split formula $\varphi = \alpha \lor \beta$, we want to obtain an interpolant circuit $C_\varphi$.

    Use the automating algorithm to find some proof of $\varphi$. Let $s_0$ be the size of such a proof. We first show that it is possible to extract a polynomial-size randomized circuit that computes the interpolant with one-sided error.
    Consider the circuit that takes as input the restriction $\rho$ together with some random bits and proceeds to compute $C_{|\alpha|, s_0^d}(\alpha_{\restriction \rho}, r)$. If this circuit finds a proof of $\alpha_{\restriction \rho}$ and it is checked to be correct, we output $0$; else, we output $1$. We claim that for at least $2/3$ choices of $r$, this circuit is a correct interpolant (and, in fact, whenever it outputs $0$, it is always correct). First, note that if we output $0$ it is because a proof of $\alpha_{\restriction \rho}$ was found, in which case it is correct to say that $\alpha_{\restriction \rho}$ is a tautology. Otherwise, we will always output $1$. The only problematic case is when the circuit outputs $1$ while $\neg \beta_{\restriction \rho}$ is satisfiable. If such was the case, then let $\sigma$ be a satisfying assignment to the $z$-variables such that $\neg \beta_{\restriction \rho, \sigma}$ is satisfied. Since $S$ can prove $\varphi$ in size $s_0$ and $S$ is closed under restrictions, we know that $S$ can prove $\varphi_{\restriction \rho, \sigma}$ in size $s_0^{d}$, and this proof must clearly be deriving $\alpha_{\restriction \rho, \sigma} = \alpha_{\restriction \rho}$. Since $s_0^{cd} \geq s_0^d$, for a \say{good} choice of $r$ the circuit $C_{|\alpha|, s_0^d}(\alpha_{\restriction \rho}, r)$ would have found such a proof, so the only reason why we could have output $1$ is that we chose a bad $r$. But this of course only happens with probability at most $1/3$. So this randomized circuit interpolates $\varphi$, makes only one-sided error, and has size polynomial in the size of the shortest proof.

    We now replicate the strategy used in Adleman's theorem ($\BPP \subseteq \P/\poly$) to show that in fact randomness is not needed in the circuit. One can follow here the standard argument as presented, for example, by Arora and Barak \cite[Thm. 7.15]{AB09}: given the interpolant circuit $F_\varphi$, perform error reduction and then argue that there must be a string of random bits that is \say{good} for all inputs of the same size. The circuit no longer makes mistakes and computes $f_\varphi$ as desired.
\end{proof}

\begin{remark}[Constructive feasible interpolation]
\label{rem:constructive}
    Our definition of feasible interpolation deviates from the one given in standard texts like that of \citeauthor{krajicekBOOK} \cite{krajicekBOOK}, and follows instead the one given by \citeauthor{pudlak03} \cite{pudlak03}, who imposes the condition that the interpolant circuit must be constructed from the given proof in polynomial time. Note that even if we adopted the non-constructive definition, the kind of feasible interpolation obtained by the construction above achieves this property anyway.

    The constructivity requirement is useful to obtain a sort of converse of Impagliazzo's observation: if a propositional proof system has uniform polynomial-size proofs of its reflection principle, then it is weakly automatable (see \autocite[Prop. 3.6]{pudlak03}).
\end{remark}

Since randomness does not buy us anything when it comes to proof search, all hardness results immediately transfer to the randomized setting. In particular, for every proof system $S$ simulating $\TC^0$-Frege, $S$ is not weakly random automatable unless Blum integers can be factored by polynomial-size randomized circuits. For weak proof systems where automatability is known to be $\NP$-hard, the systems cannot be automatable unless $\NP \subseteq \BPP$.

When moving to the quantum setting, unfortunately, we do not know of any way to get a deterministic circuit for the interpolant. Instead, we have the following natural version of Impagliazzo's observation.

\begin{proposition}
\label{prop:quant-impa}
    If a proof system is quantum automatable and closed under restrictions, then it admits feasible interpolation by quantum circuits. 
\end{proposition}

\begin{proof}
    The proof follows the argument in \Cref{prop:rand-impa}, except we can no longer apply the final step to get rid of quantumness. The interpolant now is a quantum circuit, since it is simulating the quantum circuit $C_{|\alpha|, {s_0^{d}}}(\alpha_{\restriction \rho})$.
\end{proof}

\section{$\TC^0$-Frege is hard to quantum automate}
\label{sec:outline}

The quantum version of Impagliazzo's observation (\Cref{prop:quant-impa}) is the main tool needed for our hardness results, which we are now ready to state formally.

\begin{theorem}[Main theorem]
\label{thm:main}
    There is a constant $d_0 \in \bN$ such that for every $d \geq d_0$, if there exists a polynomial-time quantum algorithm that weakly automates $\TC^0_d$-Frege, then the LWE assumption (\Cref{ass:lwe}) is broken by a $\P$-uniform family of polynomial-size quantum circuits. Furthermore, if the weak automating algorithm is classical, the LWE assumption is broken by a uniform family of polynomial-size Boolean circuits.
\end{theorem}

We can then extend the result to $\AC^0$-Frege under a stronger assumption. This is done by applying the fact that $\TC^0$-Frege proofs can be translated into $\AC^0$-Frege proofs of subexponential size (see, for example, Theorems 2.5.6 and 18.7.3 in \cite{krajicekBOOK} or the original work on the non-automatability of $\AC^0$-Frege \cite{BDGMP04}).

\begin{corollary}
\label{cor:main-AC0}
    There is a constant $d_0 \in \bN$ such that for every $d \geq d_0$, if there exists exists a polynomial-time (quantum) algorithm that weakly automates $\AC^0_d$-Frege, then the LWE assumption is broken by a $\P$-uniform family of (quantum) circuits of size $2^{n^{o(1)}}$.
\end{corollary}

We devote the rest of the paper to formally proving \Cref{thm:main}.

Suppose $h : \{0,1\}^n \to \{0,1\}^n$ is an injective and secure one-way function. Let $x$, $y$ and $z$ denote variables ranging over $\{0,1\}^n$ and assume that $\TC^0$-Frege is able to state and refute efficiently the following unsatisfiable formula,
\[ (h(x) = z \land x_1 = 0) \land (h(y) = z \land y_1 = 1)\,, \]
where $x_1, y_1$ are respectively the first bit of $x$ and $y$. The unsatisfiability follows precisely from the fact that $h$ is injective, and hence every output has a unique preimage.

If $\TC^0$-Frege admits feasible interpolation, we are guaranteed the existence of a small circuit $C(z)$ such that
\begin{equation*}
    C(z) = \begin{cases}
        0 &\text{ if } h(x) = z \land x_1 = 0 \text{ is unsatisfiable} \\
        1 &\text{ if } h(y) = z \land y_1 = 1 \text{ is unsatisfiable}
    \end{cases}
\end{equation*}
meaning that $C$ is able to invert one bit of $z$. Since every output has a unique preimage, we can iterate the process to get the entire input string. This contradicts the assumption that $h$ is one-way.

In order to instantiate the proof strategy to rule out quantum feasible interpolation, we now need a candidate one-way function that is injective and conjectured to be post-quantum secure and for which injectivity can be proven inside the proof system. Unfortunately, to the best of our knowledge, no such candidate function is currently known, or not with enough security guarantees\footnote{In a previous version of this work we formalized the injectivity of several group-based post-quantum cryptographic assumptions, such as MOBS \cite{MOBS}, as well as variants of supersingular isogeny-based Diffie-Hellman protocols, which unfortunately all happen to be now broken more or less efficiently.}. Alternatively, we may use other cryptographic objects that do achieve some form of injectivity, such as bit commitments, but the formalization of the latter does not seem simpler than the approach we follow instead. We now explain how we avoid this issue.

The most reliable post-quantum cryptographic assumptions have their security based on worst-case reductions to lattice problems conjectured to be hard. This is the case of the Learning with Errors framework \cite{regev2009lattices}, on which we base the security of the following class of candidate one-way functions. For these functions, as well as the basic properties of them that we employ, we follow the treatment of \citeauthor{Micciancio2011} \cite{Micciancio2011}. We include the details for the proof complexity readers, who may not be familiar with these constructions.

\begin{definition}[The candidate functions $f_A$]
\label{def:fA}
Let $n\in \bN$, $m = n^{O(1)}$, $q \leq 2^{n^{O(1)}}$, and $c = {\alpha q}/{\sqrt{n}}$, where $\alpha \in [0,1]$. For every matrix $A \in \bZ_q^{m \times n}$, we define the function $f_A:  \mathbb{Z}_q^n \times  \{\varepsilon \in \mathbb{Z}_q^m : |\varepsilon| \leq 10c\sqrt{mn}\} \rightarrow \mathbb{Z}_q^m$ as
\[ f_A(s,\varepsilon)  \coloneqq  (As + \varepsilon) \bmod q\,. \]

\end{definition}

At this point, we would like to show inside $\TC^0$-Frege that the conjunction
\begin{equation}
\label{eq:nice}
   (f_A(x) = z \land x_1 = 0) \land (f_A(y) = z \land y_1 = 1)
\end{equation}
is a contradiction, where $A$ is represented by free variables and $x_1$ and $y_1$ refer to the first bits of $x$ and $y$. Unfortunately, the formula is not necessarily a contradiction, since for some choices of $A$, the function $f_A$ is not injective. We can show, however, that with high probability over the choice of $A$, the function $f_A$ will satisfy two conditions that imply injectivity. Namely, $A$ will be full rank and the shortest vector in the $q$-ary lattice spanned by $A$ will be large enough.

The following proposition, which captures this idea, is standard. We reprove it here for the sake of completeness, since we shall formalize part of it inside the proof systems later.

\begin{proposition}
\label{prop:injectivity-everywhere}
    Let $n \in \mathbb{N}, m = n\log n$ and $q \geq n^{5}$. With high probability over the choice of $A \in \bZ^{m\times n}_q$, $\rank(A) = n$ and $\lambda_1(\LL_q(A)) > 20c\sqrt{nm}$. Furthermore, when these hold, the function $f_A$ is injective.
\end{proposition}

\begin{proof}
From \Cref{lem:random-lattices}.i we get that $\Pr_{A\sim \mathcal{U}(\bZ_q^{m \times n}) }[\rank(A) < n] \leq {n}/{q^{m-n+1}}$ . By \Cref{lem:random-lattices}.ii we can see that
        \[ \Pr_{A\sim \mathcal{U}(\bZ_q^{m \times n})}[\lambda_1(\LL(A)) \leq 20c\sqrt{mn} \mid \rank(A) = n] \leq \frac{(40c\sqrt{mn}+1)^m}{q^{m-2n-1}}\,. \]
The probability that a random $A$ does not satisfy the conditions in the statement is at most the sum of the two probabilities above, which are both negligible for our choice of $m$ and $q$.

For injectivity, suppose for contradiction that there exist $x,x',\varepsilon,\varepsilon'$, with either $x \neq x'$ or $\varepsilon \neq \varepsilon'$, causing a collision $f_A(x, \varepsilon) = Ax+\varepsilon = Ax' + \varepsilon' = f_A(x', \varepsilon')$. We have two cases.
    \begin{enumerate}
        \item[(a)] If $\varepsilon=\varepsilon'$, then the collision happens if and only if $\rank(A) < n$, which contradicts the assumption.
        
        \item[(b)] Suppose that $\varepsilon\neq\varepsilon'$. We have that $\varepsilon-\varepsilon' = A(x'-x)$. Since the norm of $\varepsilon-\varepsilon'$ is at most $20c\sqrt{nm}$, by transitivity we have that the length of $A(x'-x)$ is bounded by the same quantity. However, the latter belongs to the lattice and therefore we obtain a contradiction.  
        \qedhere
    \end{enumerate}
\end{proof}

These two conditions turns out to be succinctly certifiable! Indeed, to certify that the matrix $A$ is full rank we may provide a left-inverse $A^{-1}_L$ such that $A^{-1}_LA = I_n$. Unfortunately, we cannot guarantee that all injective $f_A$ have simple certificates of the second property, $\lambda_1(\LL_q(A)) > 20c\sqrt{nm}$. Nevertheless, we show in \Cref{subsec:existence-cert} that almost all of them do. These certificates take the form of sets $W = \{ w_1, \dots, w_m\} \subseteq \Delta_q^*(A)$ of short linearly independent vectors in the dual of the $q$-ary lattice. We prove---using the left inequality of \citeauthor{banaszczyk1993new}'s Transference Theorem---that such a set suffices to certify the second property, and then show---using the right side of \citeauthor{banaszczyk1993new}'s Transference Theorem--- that the certificate $W$ exists with high probability.

\begin{definition}[Certificate of injectivity]
\label{def:cert-inj}
    A \emph{certificate of injectivity} for the function $f_A$, with $A \in \mathbb{Z}_q^{m \times n}$, is a pair $(A^{-1}_L, W)$ such that $A^{-1}_L$ is a left-inverse so that $A^{-1}_L A = I_n$, and $W = \{ w_1, \dots, w_m\} \subseteq \Delta_q^*(A)$ is a set of $m$ linearly independent vectors such that $\max_{i \in [m]} ||w_i|| <  1/20c\sqrt{nm}$.
\end{definition}

The relation between injectivity and these certificates is made formal as follows.

\begin{proposition}
\label{prop:W-equivalence}
    Let $n \in \mathbb{N}, m = n\log n$, $q = n^{5}$, and $A \in \bZ^{m\times n}_q$. The following hold:
    \begin{enumerate}
        \item[(i)] if there is a certificate of injectivity $(A^{-1}_L, W)$ for $f_A$, then $f_A$ is injective;
        \item[(ii)] if $\rank(A) = n$ and $\lambda_1(\LL_q(A)) > 20mc\sqrt{nm}$, then there exists a certificate of injectivity for $f_A$;
        \item[(iii)] with high probability over the choice of $A$, $\rank(A) = n$ and $\lambda_1(\LL_q(A)) > 20mc\sqrt{nm}$.
    \end{enumerate}
\end{proposition}

Observe that given a certificate $(A^{-1}_L, W)$, verifying its correctness is a rather simple task: it is sufficient to check that $A^{-1}_L A = I_n$, to verify that  $W$ is a set of linearly independent vectors in $\Delta_q^*(A)$, and finally to ensure that the vectors in $W$ are small enough.

Let us return to the propositional system. We denote by $\operatorname{\textsc{Inj}}(f_{A})$ the propositional formula encoding that $f_A$ is injective. From this formula, $\TC^0$-Frege can derive that (\ref{eq:nice}) is a contradiction. However, $\operatorname{\textsc{Inj}}(f_{A})$ is false if we leave $A$ as free variables. We instead prove $\operatorname{\textsc{Inj}}(f_{A_0})$ for concrete injective $f_{A_0}$, where $A_0$ is hardwired. The concrete $f_{A_0}$ for which we do it are the ones that admit a certificate of injectivity.

Essentially, we formalize inside $\TC^0$-Frege that a certificate of injectivity implies injectivity. That is,
\begin{equation}
\label{eq:implication}
    \TC^0\text{-Frege} \vdash \operatorname{\textsc{Cert}}(C_A) \to \operatorname{\textsc{Inj}}(f_A)\,,
\end{equation}
where $\operatorname{\textsc{Cert}}(C_A)$ encodes that $C_A$ is a correct certificate for $f_A$. Here $C_A$ and $A$ are free variables. This implication is precisely \Cref{prop:W-equivalence}.i above. The proof inside the system is carried out in \Cref{subsec:formalization}.

Now, given a concrete certificate $C_{A_0}$ for $f_{A_0}$, the formula $\operatorname{\textsc{Cert}}(C_{A_0})$ is derivable inside $\TC^0$-Frege, which amounts to the system verifying the certificate's correctness. From this, $\TC^0$-Frege proves $\operatorname{\textsc{Inj}}(f_{A_0})$.

The rest of this section completes the missing parts in the proof. \Cref{subsec:security} sketches the known fact that $f_A$ is worst-case one-way based on the hardness of Learning with Errors, while \Cref{subsec:existence-cert} proves \Cref{prop:W-equivalence} showing the existence of certificates. We remark that the arguments and techniques are standard in cryptography and readers familiar with the area might want to skip them. We include them for the sake of completeness and to cater to the proof complexity reader that may have never come across these ideas before, and we refer to standard texts like \cite{Micciancio2011} for further details. Finally, \Cref{subsec:formalization} formalizes the certificate-to-injectivity implication above inside the theory $\LAQ$, which propositionally translates into $\TC^0$-Frege. \Cref{subsec:proof-main} reconstructs the final argument.

\subsection{Security of $f_A$}
\label{subsec:security}

The functions in $\{ f_A\}_{A \in \bZ_q^{m \times n}}$ very closely resemble the standard Learning with Errors functions, the only difference being that we have set a maximum value on the magnitude of the error vectors and allowed these to be chosen as a uniform part of the input (instead of being sampled from a Gaussian distribution). We now observe that inverting these functions allows us to invert LWE with high probability over the choice of the error vector.

\begin{lemma}[{\cite[Section 3.2]{Micciancio2011}}]
\label{lme:security}
Suppose there exists an algorithm $B$ taking as input $A\in \bZ_q^{m \times n}$ and a string $z$ and outputting a preimage in $f^{-1}_A(z)$ with probability $p$. Then, LWE can be broken with probability $0.99p$ over the choice of the error vector $\varepsilon$ and the internal randomness of $B$.
\end{lemma}
\begin{proof}[Proof sketch]
It suffices to show that with high probability the error vectors in the standard Learning with Errors functions are bounded as in our definition of $f_A$, and thus the same inverter for $f_A$ will also work for most of the original LWE instances. This follows from a standard Gaussian tail bound. Thus, if we are able to invert $f_A$ on all outputs with probability $p$, then we are able to invert its corresponding LWE function with probability, say, $0.99p$ over the choice of $\varepsilon$.
\end{proof}

Note that it is in fact possible to invert with all but negligible probability, since finding a vector whose norm is far above the expectation with high probability requires that several independently sampled coordinates all return values much larger than the expected one. For simplicity, we use this weaker result which suffices for our applications.

\subsection{Existence of certificates of injectivity: Proof of \Cref{prop:W-equivalence}}
\label{subsec:existence-cert}

This section proves the three statements of \Cref{prop:W-equivalence}. 

\begin{customprop}{\ref{prop:W-equivalence}.i}[Correctness of certificates]
If there is a certificate of injectivity $(A^{-1}_L, W)$ for $f_A$, then $\rank(A)=n$ and $\lambda_1(\LL_q(A)) > 20c\sqrt{nm}$, and thus the function $f_A$ is injective.
\end{customprop}
\begin{proof}
    As discussed in the proof of \Cref{prop:injectivity-everywhere}, $f_A$ is injective if and only if both $\rank(A) = n$ and $\lambda_1(\LL_q(A)) > 20c\sqrt{mn}$. By elementary linear algebra, $\rank(A) = n$ if and only if there exists a left-inverse $A_L^{-1}$.  As previously observed we know that $\lambda_1(\Delta_q(A)) = \min(q,\lambda_1(\LL(A)))$ and since $20c\sqrt{mn} \leq q/m$, therefore it suffices to show that the existence of $W$ as described above implies that $\lambda_1(\Delta_q(A)) > 20c\sqrt{mn}$.

    Because it is a $q$-ary lattice we known that $\rank(\Delta_q(A)) = m$. By rearranging the left inequality of the Transference Theorem for rank-$m$ lattices, we get that $\lambda_1(\Delta_q(A)) \geq {1}/{\lambda_m(\Delta_q^*(A))}$. By the definition of $W$, we conclude that ${\lambda_m(\LL^*)} < 1/20c\sqrt{nm}$, which implies that $\lambda_1(\LL_q(A)) =\lambda_1(\Delta_q(A)) > 20c\sqrt{nm}$.

    Injectivity of $f_A$ now immediately follows from the argument in \Cref{prop:injectivity-everywhere}.
\end{proof}

\begin{customprop}{\ref{prop:W-equivalence}.ii}[Conditional existence of certificates]
If $\rank(A) = n$ and $\lambda_1(\LL_q(A)) > 20mc\sqrt{nm}$, then there exists a certificate of injectivity for $f_A$.
\end{customprop}
\begin{proof}
    Since we assumed that $\rank(A) = n$, there exists a left inverse $A_L^{-1}$ for $A$. It therefore suffices to show that if $\rank(A) = n$ and $\lambda_1(\LL_q(A)) > 20mc\sqrt{nm}$, then there exists a set of vectors $W$ satisfying the conditions above.

    By the right inequality of the Transference Theorem for rank-$m$ lattices, we can obtain that ${\lambda_m(\Delta_q^*(A))} \leq m/\lambda_1(\Delta_q(A))$. Since $\lambda_1(\LL_q(A)) > 20mc\sqrt{nm}  \leq q$, and $\lambda_1(\Delta_q(A)) = \min(q,\lambda_1(\LL_q(A))) = \lambda_1(\LL_q(A))$ there must exist a set of $m$ linearly independent vectors in $\Delta_q^*(A)$, such that $\max_{i \in [m]} ||w_i||<1/20c\sqrt{nm}$.
\end{proof}

\begin{customprop}{\ref{prop:W-equivalence}.iii}[Existence of certificates with high probability]

Let $n \in \mathbb{N}, m = n\log n$, $q \geq n^{5}$, $c \leq {\sqrt{nm}}/{40}$ and $A \in \bZ_q^{m\times n}$ be sampled uniformly at random. The probability that $\rank(A) = n$ and $\lambda_1(\LL_q(A)) > 20cm\sqrt{nm}$ is at least    
    \[ 1-\frac{n}{q^{m-n+1}}- m^{3m - ( m-2n-1)\log_m q} \text{.} \]
    This probability is at least exponentially close to $1$ for our choice of $q$ and $m$.
\end{customprop}
\begin{proof}

For the following equations we define $\textsc{E}_{\text{short}}$ to be the event that $\lambda_1(\LL_q(A)) \leq 20cm\sqrt{nm}$. We have that
\begin{align*}
\Pr_A[\rank(A) \neq n \lor \textsc{E}_{\text{short}}] & = \Pr_A[\rank(A) \neq n] + \Pr_A[\textsc{E}_{\text{short}} \land \rank(A) = n]\\
& \leq  \Pr_A[\rank(A) \neq n] + \Pr_A[\textsc{E}_{\text{short}} \mid \rank(A) = n]
\,.
\end{align*}

By \Cref{lem:random-lattices}.i, we know that $\Pr_A[\rank(A) \neq n] \leq n/q^{m-n+1}$, and by the second point of \Cref{lem:random-lattices}.ii, we have that
    \[ \Pr[\textsc{E}_{\text{short}} \mid \rank(A) = n] \leq \frac{(40mc\sqrt{mn})^m}{q^{m-2n-1}} 
    \leq \frac{(m^2n)^m}{q^{m-2n-1}} 
    \leq \frac{m^{3m}}{m^{(m-2n-1)\log_m q }} 
    =  m^{3m - (m-2n-1)\log_m q}\,. \]
\end{proof}

\subsection{Formalization}
\label{subsec:formalization}
At this point, the only thing left is the formalization of the implication  $\operatorname{\textsc{Cert}}(C_A) \to \operatorname{\textsc{Inj}}(f_A)$ inside the propositional system. Since this is rather cumbersome, we work instead in the more convenient theory $\LA$ of linear algebra of \citeauthor{cooksoltys} \cite{cooksoltys}. The theory, however, is field-independent, which means we cannot state or prove properties about the ordering of the rationals, which is needed in our arguments. Furthermore, we sometimes use the fact that certain matrices are over the integers, so we must be able to identify certain elements as integers. To solve this, we introduce a conservative extension of the theory, called $\LAQ$, which assumes the underlying field to be $\bQ$.

\subsubsection{The conservative extension $\LAQ$}
\label{subsubsec:LAQ}
On top of the existing symbols of the language of $\LA$, we have two new predicate symbols $\mathsf{int}$ and $<_\bQ$. The $\LAint$ predicate, applied to a field element $q$, written $\LAint(q)$, is supposed to be true whenever the rational $q$ is an integer.

The symbol $<_\bQ$, which we overload onto $<$ in what follows, is intended to represent the usual ordering relation over the rationals. For convenience, we also add the symbol $x \leq y$ together with an axiom imposing that its meaning is $x < y \lor x = y$. Recall that equality of field elements was a symbol in the base theory $\LA$, which already equipped it with its corresponding axioms.

We now extend the axiom-set of $\LA$ with axioms for the new symbols. Recall that the original axioms of $\LA$ are listed in \Cref{app:LA-axioms}.

\begin{multicols}{2}
\begin{enumerate}
    {\item[\phantom{1}] \textbf{Axioms for $\LAint$}
        \begin{enumerate}
            \item[($\text{Int}_1$)] $\LAint(0)$
            \item[($\text{Int}_2$)] $\LAint(1)$
            \item[($\text{Int}_3$)] $\LAint(-1)$
            \item[($\text{Int}_4$)] $\LAint(x) \land \LAint(y) \to \LAint(x + y)$ 
            \item[($\text{Int}_5$)] $\LAint(x) \land \LAint(y) \to \LAint(x \cdot y)$
            \item[($\text{Int}_6$)] $\LAint(x) \land 0 < x \to 1 \leq x$

        \end{enumerate}}
    {\item[\phantom{1}] \textbf{Axioms for $<_\bQ$}
        \begin{enumerate}
            \item[($\text{Ord}_1$)] $x \leq y \leftrightarrow (x < y \lor x = y)$
            \item[($\text{Ord}_2$)] $\neg (x < x)$
            \item[($\text{Ord}_3$)] $x < y \to \neg (x = y)$
            \item[($\text{Ord}_4$)] $x < y \land y < z \to x < z$
            \item[($\text{Ord}_5$)] $\neg(x = y) \to x < y \lor y < x$

            \item[($\text{Ord}_6$)] $x \leq y \land z \leq w \to x + z \leq y + w$
            \item[($\text{Ord}_{6'}$)] $x < y \land z < w \to x + z < y + w$
            \item[($\text{Ord}_7$)] $0 \leq x \land 0 \leq y \to 0 \leq x \cdot y$
            \item[($\text{Ord}_8$)]  $0 \leq x \cdot x$
            \item[($\text{Ord}_9$)]  $0 \leq x \land y < 0 \to x \cdot y \leq 0$
            \item[($\text{Ord}_{10}$)]  $a, b, c, d \geq 0 \land a < b \land c < d \to ac < bd$
        \end{enumerate}}
\end{enumerate}
\end{multicols}

The axioms for the ordering symbols are essentially the axioms of a strict total order ($\text{Ord}_2$-$\text{Ord}_5$), together with an axiom connecting $\leq$ and $<$ ($\text{Ord}_1$). We then ensure the compatibility of the operations with the ordering relation, ($\text{Ord}_6$-$\text{Ord}_{10}$). Our axioms are not necessarily minimal, since we are interested in convenience rather than succinctness. 

The axioms for $\LAint$ are more ad hoc and it might seem that they are not enough to fix the correct interpretation of the symbol. Indeed, the axioms for $\LAint$ only really force that addition and multiplication are closed under this predicate and that every integer in the standard model can be argued to be an integer in $\LAQ$, but they do not necessarily identify $\bZ$ as a substructure of $\bQ$. The reason this is not an issue is that we are only interested in $\LAQ$ for its propositional translation. For our purposes, these axioms are the only ones we need to prove the required claims about lattices, and once we translate to the propositional setting, the symbols will take the standard intended interpretation, so it does not matter that these are underspecified in the theory.

It is not hard to show that theorems of $\LAQ$ admit succinct $\TC^0$-Frege proofs. Showing this requires extending the propositional translation of \citeauthor{cooksoltys} to include the new symbols and axioms. We do this in \Cref{app:LAQ-trans}.

\subsubsection{Formalization of the proofs}
\label{subsubsec:formal-final}
We are ready to present the formal proofs needed inside our theory. In what follows, LA\_.\_ stands for the corresponding axiom in \Cref{app:LA-axioms}. 

First, we observe that under our new axioms, $\LAQ$ can argue that the inner product of a vector with itself is non-negative. Recall that the inner product operator $\langle u, v \rangle$ is a defined term in $\LA$, namely $u\cdot v^\intercal$.

\begin{lemma}
\label{lemma:inner}
    Provably in $\LAQ$, for every $v\in\bQ^n$, $ 0 \leq \langle v, v \rangle$. Furthermore, if $v \neq 0$, then $0 < \langle v, v \rangle$.
\end{lemma}

\begin{proof}
    Unfolding the definition of the term $\langle v, v \rangle$ in $\LAQ$, $\langle v, v \rangle = \sum_{i=1}^n v_i \cdot v_i$,
    so by axiom ($\text{Ord}_8$) each term in the sum satisfies $v_i \cdot v_i \geq 0$, and by repeatedly applying axiom ($\text{Ord}_6$), the entire sum can be proven to be non-negative. Here we are implicitly using the Induction Rule of $\LA$ to formalize the argument (see \Cref{app:inf-rules}). Now we show that non-zero elements have non-zero inner product: if $v > 0$, that means each entry $v_i$ satisfies $v_i > 0$; by ($\text{Ord}_{10}$) we have that $v_i \cdot v_i > 0$. The Induction Rule together with ($\text{Ord}_{6'}$) completes the argument.
\end{proof}

We now formalize inside $\LAQ$ the classical Cauchy-Schwartz inequality.

\begin{lemma}[Cauchy-Schwartz in $\LAQ$]
\label{lem:CS}
The theory $\LAQ$ proves that for every $u,v\in\mathbb{Q}^n$, $\langle u,v\rangle^2\leq\langle u,u\rangle\cdot\langle v,v\rangle$.
\end{lemma}

\begin{proof}
If $v$ is the all-zeroes vector, then the inequality holds trivially, so let us assume $v > 0$.
We first show that $\LAQ$ can derive the following equality, 
\begin{equation}
\label{eq:CS}
        \frac{1}{\langle v,v\rangle}\langle(\langle v,v\rangle u-\langle u,v\rangle v), (\langle v,v\rangle u -\langle  u,v\rangle v)\rangle=\langle u,u\rangle\langle v ,v\rangle - \langle u,v\rangle^2 \,,
\end{equation}
where $\frac{1}{\langle v,v\rangle}$ can be explicitly referred to in $\LAQ$ as $\langle v,v\rangle^{-1}$.

We do this by deriving the following chain of equalities,
\begin{align}
       \frac{1}{\langle v ,v\rangle}\langle(\langle v ,v\rangle u-\langle  u,v\rangle v), (\langle v ,v\rangle u-\langle  u,v\rangle v)\rangle &=  \tag{\text{LA7.b}} \\   
     \frac{1}{\langle v,v\rangle}(\langle(\langle v ,v\rangle u-\langle  u,v\rangle v), (\langle v ,v\rangle u)\rangle + \langle(\langle v ,v\rangle u-\langle  u,v\rangle v), (-\langle  u,v\rangle v)\rangle) &= \tag{\text{LA7.b + LA7.c + LA7.a}} \\
\frac{1}{\langle v,v\rangle}(\langle v,v\rangle^2 \langle  u, u\rangle - \langle v,v\rangle\langle v ,u\rangle^2) &= \tag{\text{LA7.c}} \\
 \langle v,v\rangle \langle u, u\rangle - \langle v,u\rangle^2 \notag.
\end{align}

Observe now that from \Cref{lemma:inner} above, we know that $\langle(\langle v,v\rangle u-\langle u,v\rangle v), (\langle v,v\rangle u -\langle  u,v\rangle v)\rangle$ is non-negative. Furthermore, $ \frac{1}{\langle v ,v\rangle}$ is also non-negative (and, in fact, strictly positive): by \Cref{lemma:inner}, since $v > 0$ , we have $\langle v ,v\rangle > 0$. It now suffices to argue that for any field element $a$, if $a >0$, then $a^{-1} > 0$. Indeed, since $a >0$ and thus $a \neq 0$, axiom LA3.d guarantees that $a \cdot a^{-1} = 1$. If $a^{-1} \leq 0$ then we get a contradiction: either $a^{-1} = 0$, which contradicts $aa^{-1} = 1$, or $a^{-1} < 0$; in the latter case, by axiom $(\text{Ord}_9)$, $aa^{-1} \leq 0$, contradicting again $aa^{-1} = 1$. This argument applied to $\langle v, v \rangle > 0$ gives us  $\langle v ,v\rangle^{-1} > 0$.

Since both factors are non-negative, by $(\text{Ord}_7)$, their product is also non-negative, meaning that
\begin{equation*}
   0 \leq  \frac{1}{\langle v,v\rangle}\langle(\langle v,v\rangle u-\langle u,v\rangle v), (\langle v,v\rangle u -\langle  u,v\rangle v)\rangle=\langle u,u\rangle\langle v ,v\rangle - \langle u,v\rangle^2 \,,
\end{equation*}
and thus, by the transitivity axiom $(\text{Ord}_4)$, the right-hand side of \Cref{eq:CS} is non-negative as well: we have $\langle u,u\rangle\langle v ,v\rangle - \langle u,v\rangle^2 \geq 0$. Rearranging the inequality using $(\text{Ord}_6)$, the inequality follows.
\end{proof}

The other technical component needed in the final proof is a weakening of the lower bound in \citeauthor{banaszczyk1993new}'s Transference Theorem (see \Cref{thm:transference}). Informally, we need to prove that for every $A\in\mathbb{Q}^{m\times n}$, every non-zero vector $v\in\mathcal{L}(A)$ and any set of linearly independent vectors $W = \{w_1, \dots, w_n\} \subseteq \mathcal{L}^*(A)$, $\langle v ,v\rangle\cdot \langle w_i,w_i\rangle\geq 1$ for some $i \in [n]$.

In order for $\LAQ$ to process the conditions of the theorem, we provide certificate-like objects ensuring all the different hypotheses. For example, when we quantify over a vector $v$ belonging to a lattice $\LL(A)$, we provide the vector of coefficients $c_v$ such that $Ac_v = v$. Note as well that when we quantify over matrices with elements in $\bZ$, we are using the $\mathsf{int}$ predicate under the hood to enforce the entries to be integers. As a final remark, recall that we do not have existential quantifiers in $\LA$, but whenever we do some existential quantification in the following lemmas we are quantifying over small finite domains, meaning we can write everything as a small disjunction.

\begin{lemma}[\citeauthor{banaszczyk1993new}'s left inequality in $\LAQ$]
\label{lem:lwb-LAQ}
The theory $\LAQ$ proves the following implication. Let $A \in\mathbb{Z}^{m\times n}$, $B \in \bQ^{n\times n}$, $v \in \bQ^n$, $c_v \in \bZ^m$,  $W =[w_1|\dots|w_n]\in \mathbb{Q}^{m\times n}$, $c_{W} = [c_{w_1} \mid \dots \mid c_{w_n}] \in \mathbb{Z}^{m\times n}$, $W' \in \bQ^{m \times n}$ fulfilling the following conditions:

\begin{enumerate}
    \item the vector $v$ is non-zero, $v\neq 0_n$;
    \item the vector $v$ belongs to the lattice $\LL(A)$, $v = Ac_v$;
    \item the vectors in $W$ belong to the dual lattice\footnote{This dual lattice, in fact, admits a closed form for its base, as in \Cref{lem:dual-basis}. In particular, $B$ can be seen as $A(A^\intercal A)^{-1}$.} $\LL^*(A)$, $w_i = ABc_{w_i}$ for all $i \in [n]$;
    \item $(A^\intercal A)B = I_n$;
    \item the vectors in $W$ are linearly independent, $W'W^\intercal = I_n$.
\end{enumerate}
Then, for some $i \in [n]$, $\langle v ,v\rangle\cdot\langle w_i,w_i\rangle\geq 1 $.
\end{lemma}
\begin{proof}
The proof has two steps. First, we show that for all $i\in[n]$, $\langle v ,w_i\rangle\in\mathbb{Z}$. To do this we use the following chain of equalities, where $w$ is some arbitrary column $w_i$ of $W$, and where the comments on the side refer to either axioms of $\LAQ$ or the assumptions in the statement of the lemma:
\begin{align}
     \langle v ,w\rangle    &=\langle Ac_v, ABc_w\rangle \tag{\text{by ass. 2 and 3}} \\   
        &= c_v^\intercal A^\intercal  ABc_w              \tag{\text{by def. of dot product}} \\
          &=  c_v^\intercal (A^\intercal A)Bc_w \tag{\text{by associativity, LA5.i}} \\
            &= c_v^\intercal c_w \text{.} \tag{by ass. 4}
\end{align}
By assumption, the entries in both $c_v$ and $c_w$ have integer entries, so by the closure under integer multiplication and addition ($\text{Int}_4$ and $\text{Int}_5$) we have that $c_v^\intercal c_w$ is an integer, and thus we deduce that for all $i \in[n]$, $\langle v ,w_i\rangle\in\mathbb{Z}$.

In the second step of the proof, we show that there is $i \in[n]$ such that $\langle v ,w_i\rangle\neq 0$. We consider the vector $s  \coloneqq  W^\intercal v$. Note that by definition, the $j$-th entry of $s$ is $\langle w_j,v\rangle$. We can multiply both sides by the same matrix $W'$, leading to $W's=W'W^\intercal v$. Using associativity (LA5.i), assumption (5) and properties of the identity matrix (LA5.f), we get that $W's=v$. Suppose that for all $i\in[n]$, $\langle v ,w_i\rangle=0$. Then, by definition, $s=0_n$. We can easily derive (using LA3.a, LA3.c and LA3.i) that $W's=0$ and therefore $v=0$. This contradicts assumption (1). 

Finally, let $i$ denote the particular index for which we have now derived that simultaneously $\langle v, w_i \rangle \in \bZ$ and $\langle v, w_i\rangle \neq 0$. By axiom ($\text{Ord}_8$), $\langle v, w_i\rangle^2 \geq 0$. Furthermore, it is easy to derive already in $\LA$ that for any field elements $a$ and $b$, if $a\neq 0$ and $b \neq 0$, then $ab \neq 0$ (this follows immediately from axioms LA3.a-d). Thus, by axiom ($\text{Ord}_1$), $\langle v, w_i\rangle^2 > 0$. Recall now that by axiom ($\text{Int}_6$) of $\LAQ$, every non-zero positive integer is greater or equal than $1$, so $\langle v, w_i\rangle^2 \geq 1$. Then, the Cauchy-Schwartz inequality from \Cref{lem:CS} gives us
\[ 1 \leq \langle v, w_i\rangle^2 \leq \langle v, v \rangle \cdot \langle w_i, w_i \rangle,
 \]
which together with transitivity (axiom $\text{Ord}_4$ together with $\text{Ord}_1$) yields the desired $1 \leq \langle v, v \rangle \cdot \langle w_i, w_i \rangle$.
\end{proof}

We are now ready to prove in $\LAQ$ that a correct certificate of injectivity implies the injectivity of $f_A$. Informally, we aim to prove that given a certificate of injectivity as in \Cref{def:cert-inj}, the function $f_A$ is injective. However \cref{lem:lwb-LAQ} only applies to the extension lattice $\Delta_q(A)$. Fortunately proving injectivity for $\Delta_q(A)$ is sufficient for inverting $f_A$. As before, we need to provide some additional objects together with the certificate to make sure $\LAQ$ can reason about this conditional implication and carry out the verification of the certificate. 

\begin{lemma}[Certificate-implies-injectivity in $\LAQ$]
\label{lem:inj-formal}
Let $A \in\mathbb{Z}^{m\times m}$, $B \in \bQ^{m\times m}$, $v_1 \in \bQ^m$, $c_{v_1} \in \bZ^m$, $v_2\in \bQ^m$, $c_{v_2} \in \bZ^m$, $\varepsilon_1 \in \bQ^m$, $\varepsilon_2 \in \bQ^m$, $W =[w_1|\dots|w_m]\in \mathbb{Q}^{m\times m}$, $c_{W} = [c_{w_1} \mid \dots \mid c_{w_n}] \in \mathbb{Z}^{m\times m}$, $W' \in \bQ^{m \times m}$ fulfilling the following conditions:

\begin{enumerate}
    \item the vector $v_1$ belongs to the lattice $\LL(A)$, $v_1 = Ac_{v_1}$;
    \item the vector $v_2$ belongs to the lattice $\LL(A)$, $v_2 = Ac_{v_2}$;
    \item the vectors $v_1$ and $v_2$ are distinct, $v_1 \neq v_2$;
    \item the vectors in $W$ belong to the dual lattice $\LL^*(A)$, $w_i = Bc_{w_i}$ for all $i \in [n]$;
    \item $(A^\intercal A)B = I_n$;
    \item the vectors in $W$ are linearly independent, $W'W^\intercal = I_n$;
    \item $\langle w_i, w_i \rangle < 1/400c^2nm$ for all $i \in[m]$;
    \item $\langle \varepsilon_2 - \varepsilon_1, \varepsilon_2 - \varepsilon_1 \rangle < 400c^2nm$. 
\end{enumerate}

Then, $v_1+\varepsilon_1\neq v_2 +\varepsilon_2$.

\end{lemma}
\begin{proof}
The proof proceeds by contradiction. Suppose that $Av_1+\varepsilon_1= Av_2+\varepsilon_2$, meaning that a collision exists in the range of $f_A$. By simple algebraic manipulations in $\LAQ$, we derive that $A(v_1-v_2)=\varepsilon_2-\varepsilon_1$. Let $v \coloneqq A(v_1-v_2)$ and $\varepsilon \coloneqq \varepsilon_2 - \varepsilon_1 $, so that we have $v = \varepsilon$. 

Observe now that assumptions (7) and (8), together with axiom ($\text{Ord}_{10}$) and (LA3.d) give us $\langle w_i, w_i \rangle \langle \varepsilon, \varepsilon \rangle < 1$ for every $i\in[n]$. With the existing assumptions of the theorem we can in fact apply \Cref{lem:lwb-LAQ} to $v$, getting that there exists $i$ such that $\langle v,v \rangle \langle w_i, w_i \rangle \geq 1$. Since $v = \varepsilon$, we get $\langle \varepsilon,\varepsilon \rangle \langle w_i, w_i \rangle \geq 1$, but this means
\[ 1 \leq \langle \varepsilon,\varepsilon \rangle \langle w_i, w_i \rangle < 1. \]
We remark that while technically the transitivity axiom ($\text{Ord}_4$) is stated for strict orders, the existing set of axioms immediately implies the \say{mixed} version, namely that for field elements $a$, $b$ and $c$, if $a \leq b$ and $b < c$, then $a < c$. Thus we now have $1 < 1$, but this contradicts axiom ($\text{Ord}_2$).
\end{proof}

\subsection{Proof of \Cref{thm:main}}
\label{subsec:proof-main}

We are ready to put all the pieces together.

\begin{proof}[Proof of \Cref{thm:main}]
Let $d_0$ be the smallest constant for which $\TC^0$-Frege satisfies \Cref{eq:implication}, meaning that $\TC^0$-Frege admits polynomial-size proofs of the propositional translation of \Cref{lem:inj-formal} in depth $d_0$, and let $d \geq d_0$. Suppose that $\TC^0_d$-Frege is weakly quantum automatable, that is, suppose that $S$ is a quantum automatable proof system simulating $\TC^0_d$-Frege. Let $Q$ be the quantum algorithm automating $S$. We describe a quantum algorithm $Q'$ that takes as input a matrix $A$ defining a function $f_A$ as in \Cref{def:fA} and an output $z$ of this function and succeeds in finding a preimage of $z$ with high probability.

For a specific input matrix $A_0$, using \Cref{lem:q-basis} we can find $A_1$ such that $\LL(A_1) = \Delta(A_0)$. We can then consider the formula $\operatorname{\textsc{Cert}}(C_A) \to \operatorname{\textsc{Inj}}(f_A)$, where $C$ and $A$ are free variables. In \Cref{lem:inj-formal} this implication was proven inside $\LAQ$, and by the propositional translation for $\LAQ$ in \Cref{thm:new-propositional-translation} we get an efficient proof inside $\TC^0_d$-Frege, and thus also in $S$. Craft now the formula $\operatorname{\textsc{Inj}}(f_{A_1})$ for the particular $A_1$ obtained earlier. By \Cref{prop:W-equivalence}, for most $f_{A_0}$ there exists a certificate of injectivity $C_{A_1}$ such that $\operatorname{\textsc{Cert}}(C_{A_1})$ is true and, in fact, has no free variables. Consider this certificate as a partial restriction and apply it to the implication above. Since $\TC^0$-Frege systems are closed under restrictions, there must be a polynomial-size proof of $\operatorname{\textsc{Inj}}(f_{A_1})$, and so $S$ also proves this efficiently. Recall that, as noted in \Cref{rem:constructive}, Impagliazzo's observation guarantees that under the existence of an automating algorithm we get constructive feasible interpolation, so from the proof of $\operatorname{\textsc{Inj}}(f_{A_1})$ we can get a circuit that recovers one bit of $s'$ such that $A_1s' \equiv As \text{ }(\text{mod } q)$. By iterating this process we can recover the entire preimage. Taking the modulus of $s' (\text{mod } p)$ we recover $As$ which as noted after our definition of LWE suffices to win the security game. This procedure works as long as $f_{A_0}$ is injective and admits a certificate of injectivity, but by \Cref{prop:W-equivalence} this is the case with overwhelming probability. Then, by \Cref{lme:security}, we break LWE and get the desired conclusion.
\end{proof}

    The proof above is phrased from the starting assumption of a weak automating algorithm rather than feasible interpolation. The reason is that, intuitively, feasible interpolation alone does not seem to immediately break the cryptographic assumption: for every fixed matrix $A$, feasible interpolation only seems to guarantee the existence (with high probability) of a circuit breaking $f_A$, but this circuit seems to essentially depend on $A$. By starting the argument from an automating algorithm, we have a uniform way of finding the proofs of injectivity for each particular $f_A$ to then construct the corresponding interpolating circuit.

    While we find this more intuitive, we can still phrase the argument directly in terms of interpolation (and hence rule out this too under the same assumption). It suffices to argue that $\TC^0$-Frege can refute the contradictory formulas
    \[  (f_A(x) = z \land x_i = 0) \land ((f_A(y) = z \land y_i = 1 ) \land \textsc{Cert}(C_A)), \]
    where $x_i$ and $y_i$ refer to the $i$-th respective bits, and $\textsc{Cert}(C_A)$ is the certificate predicate, as in \Cref{eq:implication}. Observe that this is still a split formula, since the variables encoding the certificate $C_A$ appear only on the right-hand side. That the proof system can show this is a contradiction follows immediately from the fact that it can prove the implication in \Cref{eq:implication}. More importantly, the refutation of this formula is uniform and known, with $A$ as free variables, meaning we can extract the interpolants directly. It is not hard to see that interpolating on this formula we can still break the same functions that we would break with the aid of an automating algorithm. This remark is due to Impagliazzo. Thus, the following corollary also follows from our formalization.

\begin{corollary}
    For every large enough $d \in \bN$, if $\TC_d^0$-Frege (respectively, $\AC^0_d$-Frege) admits feasible interpolation by (quantum) circuits, then the LWE assumption can be broken by a uniform family of polynomial-size (respectively, subexponential-size) (quantum) circuits.
\end{corollary}

\subsection*{Acknowledgments}
The question of the quantum non-automatability of strong proof systems was suggested to us by three different people. We thank Vijay Ganesh for bringing it up during the Dagstuhl Seminar 22411 \emph{Theory and Practice of SAT and Combinatorial Solving}. We thank Susanna F. de Rezende for bringing our attention to the problem later and for insightful comments and careful proofreading. We would also like to thank Ján Pich for independently pointing us to the problem and for discussing many details with us. We are particularly grateful for him directing us to the work of Soltys and Cook on $\LA$, which greatly simplified our formalization. We also thank Rahul Santhanam for his insights and conversations and Yanyi Liu for pointing us to the existence of the certificates of injectivity. We also thank María Luisa Bonet, Jonas Conneryd, Ronald de Wolf, Eli Goldin, Peter Hall, Russell Impagliazzo, Erfan Khaniki, Alex Lombardi, Daniele Micciancio, Angelos Pelecanos and Michael Soltys for useful comments, suggestions and pointers.

The journal version of this work appeared in \emph{Computational Complexity} in October of 2025 \cite{ACG25_CC_journal}. A preliminary version of this work was presented at the poster session of the International Conference on Quantum Information Processing (QIP 2024) and an extended abstract appeared at the 39th Computational Complexity Conference (CCC 2024) \cite{ACG24_CCC}. We are thankful to the anonymous reviewers, particularly for observing that some crucial axioms were missing from the definition of $\LAQ$ in an earlier version of this work.

This work was done in part while the authors were visiting the Simons Institute for the Theory of Computing at UC Berkeley during the spring of 2023 for the \emph{Meta-Complexity} and \emph{Extended Reunion: Satisfiability} programs.

The first author was supported by the Wallenberg AI, Autonomous Systems and Software Program (WASP) funded by the Knut and Alice Wallenberg Foundation.

\printbibliography[
        heading=bibintoc
]

\newpage
\appendix

\section{The theories $\LA$ and $\LAQ$}

\Cref{app:LA-axioms} below lists the axioms of the theory $\LA$ of linear algebra of Soltys and Cook \cite{cooksoltys}, together with several theorems proven inside the theory in the original paper. \Cref{app:LAQ-trans} proves that the conservative extension $\LAQ$ admits a propositional translation into $\TC^0$-Frege.

\subsection{Axioms and basic theorems of $\LA$}
\label{app:LA-axioms}
\begin{multicols}{2}
\begin{enumerate}
    \item \textbf{Equality axioms}
    \begin{enumerate}
        \item $x=x$
        \item $x = y \rightarrow y = x$
        \item $(x=y \land y=z) \rightarrow x=z$
        \item $\bigwedge_{i}^n (x_i = y_i) \to f(\bar x)=f(\bar y)$
        \item $i_1=j_1, i_2=j_2, i_1\leq i_2\rightarrow j_1\leq j_2$.
    \end{enumerate}
    \item \textbf{Axioms for indices}
    \begin{enumerate} 
        \item $i+0=i$ 
        \item $ i+(j+1)=(i+j)+1$
        \item $ i\cdot(j+1)=(i\cdot j)+i$
        \item $i+1=j+1 \to i=j$
        \item $ i+1\neq 0$
        \item $ i \leq i+j$
        \item $ i\leq j, j\leq i$
        \item $i\leq j, i+k=j\to (j-i=k)$
        \item $i\leq j, i+k=j \to (i \nless j \to j-i=0)$
        \item $j \neq 0 \to \mathsf{rem}(i, j) < j$
        \item $j \neq 0 \to i = j \cdot \mathsf{div}(i,j) + \mathsf{rem}(i, j)$
        \item $\alpha \to \mathsf{cond}(\alpha, i, j)=i$
        \item $\neg \alpha\to \mathsf{cond}(\alpha, i, j)=j$
    \end{enumerate}
    \item \textbf{Axioms for field elements}
    \begin{enumerate}
        \item $0\neq 1\wedge a+0=a$  
        \item $ a+(-a)=0$
        \item $1 \cdot a = a$ 
        \item $a\neq 0\rightarrow a \cdot(a^{-1})=1$
        \item $ a+b = b+a$ 
        \item $ a\cdot b = b\cdot a$ 
        \item $ a+(b+c)=(a+b)+c$
        \item $ a\cdot (b\cdot c)=(a\cdot b)\cdot c$ 
        \item $ a\cdot (b+c)=a\cdot b+a\cdot c$ 
        \item $\alpha \to \mathsf{cond}(\alpha, a, b)=a $
        \item $\neg\alpha\to \mathsf{cond}(\alpha, a, b)=b$
    \end{enumerate}
    \item \textbf{Axioms for matrices}
    \begin{enumerate}
        \item $(i=0\vee \mathsf{r}(A)<i\vee j=0\vee \mathsf{c}(A) <j) \rightarrow \mathsf{e}(A,i,j)=0$
        \item $\mathsf{r}(A)=1, \mathsf{c}(A)=1\rightarrow\Sigma(A)=\mathsf{e}(A,1,1)$
        \item $\mathsf{c}(A)=1 \rightarrow \sigma(A)=\sigma(A^\intercal)$
        \item $\mathsf{r}(A)=0\vee \mathsf{c}(A)=0\rightarrow \Sigma(A)=0$
    \end{enumerate}
    \columnbreak
    \item \textbf{Theorems for ring properties} 
    \begin{enumerate}
        \item $\max(i,j)=\max(j,i)$
        \item $\max(i, \max(j,k))=\max(\max(i,j),k)$
        \item $\max(i,\max(j,k))=\max(\max(i,j),\max(i,k))$
        \item $A+0=A$
        \item $A+(-1)A=0$ 
        \item $AI=A \text{ and } IA=A$ 
        \item $A+B=B+A$ 
        \item $A+(B+C)=(A+B)+C$ 
        \item $A(BC)=(AB)C$ 
        \item $A(B+C)=AB+CA$ 
        \item $(B+C)A=BA+CA$  
        \item $\Sigma 0 =0_{\text{field}}$
        \item $\Sigma(cA)=c\Sigma(A)$
        \item $\Sigma(A+B)=\Sigma(A)+\Sigma(B)$
        \item $\Sigma(A)=\Sigma(A^\intercal)$
    \end{enumerate}
    \item \textbf{Theorems for module properties}
    \begin{enumerate}
        \item $(a+b)A=aA+bA$ 
        \item $a(A+B)=aA+aB$ 
        \item $(ab)A=a(bA)$
    \end{enumerate}
    \item \textbf{Theorems for inner product}
    \begin{enumerate}
        \item $A\cdot B= B\cdot A$ 
        \item $A\cdot (B+C)=A\cdot B+ A\cdot C$ 
        \item $aA\cdot B=a(A\cdot B)$
    \end{enumerate}
    \item \textbf{Miscellaneous theorems}
    \begin{enumerate}
        \item $a(AB)=(aA)B\wedge (aA)B=A(aB)$
        \item $(AB)^\intercal=B^\intercal A^\intercal$
        \item $I^\intercal=I$
        \item $0^\intercal=0$
        \item $(A^\intercal)^\intercal=A$ 
    \end{enumerate}
\end{enumerate}
\end{multicols}

\subsection{Inference rules of $\LA$}
\label{app:inf-rules}
\citeauthor{cooksoltys} devised $\LA$ more as purpose-specific formal system than a proper first-order theory and, as such, the definition of $\LA$ consists of both the axioms above as well as inference rules on how to manipulate these. The underlying inference rules are those of Gentzen's Sequent Calculus (in its propositional version, without quantifier rules), together with two additional rules described below. In what follows, $\Gamma$ and $\Delta$ are cedents (finite sequences of formulas), which may be empty. See \autocite[Section 2.2]{cooksoltys} for further discussion of the rules.

\subsubsection*{Matrix equality rule}

\begin{prooftree}
    \AxiomC{$\Gamma \longrightarrow \Delta, \mathsf{e}(T, i, j) = \mathsf{e}(U, i, j)$}
    \AxiomC{$\Gamma \longrightarrow \Delta, \mathsf{r}(T) = \mathsf{r}(U)$}
    \AxiomC{$\Gamma \longrightarrow \Delta, \mathsf{c}(T) = \mathsf{c}(U)$}
    \TrinaryInfC{$\Gamma \longrightarrow \Delta, T = U$}
\end{prooftree}

Here, the index variables $i, j$ may not occur free in the bottom sequent, and $T$ and $U$ are any matrix terms.

\subsubsection*{Induction rule}

\begin{prooftree}
    \AxiomC{$\Gamma , \alpha(i) \longrightarrow \alpha(i+1), \Delta$}
    \UnaryInfC{$\Gamma, \alpha(0) \longrightarrow \alpha(n), \Delta$}
\end{prooftree}
Here $\alpha$ is an arbitrary formula, the index variable $i$ does not occur free in either $\Gamma$ or $\Delta$ and $n$ stands for any terms of index type.

\subsection{The propositional translation for $\LAQ$}
\label{app:LAQ-trans}
We show that theorems of $\LAQ$ still have short propositional proofs in $\TC^0$-Frege despite the presence of new symbols and axioms not present in $\LA$.

The first step in the propositional translation is the conversion of $\LA$ formulas into propositional ones. The translation is analogous to the usual propositional translations used elsewhere in bounded arithmetic (see, for example, \cite{cookBOOK, krajicek95}). Let $\varphi$ be a formula of $\LA_\bQ$ and let $\sigma$ be an object assignment that assigns a natural number to each free index variable occurring in $\varphi$ and to each term of the form $\mathsf{r}(A)$ and $\mathsf{c}(A)$ occurring in $\varphi$. We denote by $N$ the maximum value in the range of $\sigma$. For every variable $q$ standing for a rational number in $\varphi$, we introduce enough Boolean variables to represent $q$ as a fraction $a/b$, where $a$ and $b$ are integers represented in binary. We may assume that $N$ is also an upper bound on the binary precision of these integers. We adopt the convention that denominators are always positive. Note that the number of Boolean variables introduced is at most polynomial in $N$.

The translation of $\varphi$ then proceeds by substituting every function and predicate symbol by the corresponding $\TC^0$ circuit of the appropriate size, which is also at most $\poly(N)$. It is easy to verify that all the functions and predicate symbols in $\LA$ are computable in $\TC^0$, and this is also the case for the extended vocabulary ($<_\bQ$ and $\LAint$). For $<_\bQ$, given $a/b$ and $c/d$, we check whether $ac < bd$, which only requires operations over the integers. For $\LAint$ we shall use the circuit computing whether the rational $q = a/b$ satisfies that $\operatorname{MOD}(a, b) = 0$. This only requires the standard remainder function, computable in $\TC^0$, plus an equality check. We denote by $||\varphi||_\sigma$ the propositional formula obtained by carrying out this translation process. The size of $||\varphi||_\sigma$ is again polynomial in $N$.

Now, given an $\LA$ proof $\pi$ of a sentence $\varphi$ and an object assignment $\sigma$, we translate $\pi$ into a $\TC^0$-Frege proof of $||\varphi||_\sigma$. It suffices to translate each line in the proof into its corresponding propositional formula and observe that the required inference steps can be carried out in $\TC^0$-Frege. More precisely, the underlying proof system for $\LA$ and $\LA_\bQ$ is the sequent calculus, and we may think of $\TC^0$-Frege in its sequent calculus formulation, $\mathsf{PTK}$. Since $\LA$ and $\LA_\bQ$ formulas are quantifier-free, no derivation rules for quantifiers are present in $\pi$. Every inference step of an $\LA$ proof matches the corresponding sequent calculus rule of the propositional sequent calculus. It is also not hard to see that the cut rule is always over a $\TC^0$ circuit, since the $\LA_\bQ$ formula over which we cut translates into a $\TC^0$ circuit.

The only problem occurs when reaching a leaf in the proof $\pi$, which corresponds to an axiom of $\LA$ or $\LA_\bQ$. These are not axioms of $\TC^0$-Frege and hence require a proof to be appended in the translation. Cook and Soltys observed that when instantiated over the rationals, all of the axioms of $\LA$ are either directly proven or follow easily from the basic properties of arithmetic proven by Bonet, Pitassi, and Raz \cite{BPR1997} inside $\TC^0$-Frege. Hence, the only thing left to complete the propositional translation for $\LA_\bQ$ is to provide small $\TC^0$-Frege proofs of the new axioms not present in $\LA$.

\begin{lemma}
\label{lem:translation-new}
    There are polynomial-size $\TC^0$-Frege proofs of the propositional translation of the axioms of $\LA_\bQ$.
\end{lemma}

\begin{proof}
    The axioms of $\LA$ were handled in the original work of Cook and Soltys (see Theorem 6.3 in \cite{cooksoltys}). Furthermore, the axioms imposing that $<_\bQ$ is an ordering relation were already proven in \cite{BPR1997} as well (these are precisely the lemmas proven in their Section 7.2). We therefore focus on the translation of the axioms for $\LAint$.
    
    For the sake of consistency with the previous work of \citeauthor{BPR1997} we adopt here the notation $[a]_b$ for the $\operatorname{MOD}(a,b)$ function and $\operatorname{div}_b(a)$ for the integer division between $a$ and $b$. We also reuse the following lemmas proved by them inside $\TC^0$-Frege, where L7.\_ stands for the corresponding lemma in \cite{BPR1997}:

    \begin{itemize}[leftmargin=1.3cm]
        \item[(L7.19)] $(a < b) \lor (b< a) \lor (a = b)$.
        \item[(L7.27)] $a = [a]_b + \operatorname{div}_b(a) \cdot b$.
        \item[(L7.28)] $ x+y\cdot p = u+ v \cdot p \land y < v \to p \leq x$.
        \item[(L7.29)] $[a]_b =[a + k\cdot b]_b$.

    \end{itemize}

    The first three axioms for $\LAint$ clearly admit constant-size proofs, so we only need to write the proofs for the axioms ($\text{Int}_4$), ($\text{Int}_5$), and ($\text{Int}_6$) 
    \begin{enumerate}[leftmargin=1.15cm]
            \item[($\text{Int}_4$)]
            
            Let $x$ and $y$ be represented by the fractions $a/b$ and $c/d$ respectively. The translation of the axiom
            \[ \LAint(x) \land \LAint(y) \to \LAint(x + y) \]
            yields the propositional formula
            \[ [a]_b = 0 \land [c]_d = 0 \to [ad + bc]_{bd} = 0\text{.} \] 

            From $[a]_b = 0$ and $[c]_d = 0$, L7.27 gives us that $a = \operatorname{div}_b(a)\cdot b$ and $c = \operatorname{div}_d(c)\cdot d$. Then,
                \begin{align*}
                    [ad + bc]_{bd} &= [{\underbrace{\operatorname{div}_b(a)\cdot b}_{a}}  \cdot d + {\underbrace{\operatorname{div}_d(c)\cdot d}_{c}}\cdot b]_{bd} \\
                    &=[bd \cdot (\operatorname{div}_b(a) + \operatorname{div}_d(c))]_{bd} \\
                    &= [0]_{bd} \\
                    &= 0
                \end{align*}
            where the second to last equality follows by applying L7.29.
            
            \item[($\text{Int}_5$)] In this case the translation of
            \[ \LAint(x) \land \LAint(y) \to \LAint(x \cdot y) \]
            yields the formula
            \[ [a]_b = 0 \land [c]_d = 0 \to [ac]_{bd} = 0\text{.} \]
            We have again that L7.27 gives us that $a = \operatorname{div}_b(a)\cdot b$ and $c = \operatorname{div}_d(c)\cdot d$. Then,
            \begin{align*}
                    [ac]_{bd} &= [ac + (-\operatorname{div}_b(a) \cdot \operatorname{div}_d(c)) \cdot bd]_{bd} \\
                    &= [ac - {\underbrace{\operatorname{div}_b(a) \cdot b}_{a}} \cdot {\underbrace{\operatorname{div}_d(c) \cdot d}_{c}}]_{bd} \\
                    &= [ac - ac]_{bd} \\
                    &= [0]_{bd} \\
                    &= 0
                \end{align*}
            where the first equality follows again from L7.29.
            
            \item[($\text{Int}_6$)] We first write the propositional translation of 
            \[ \LAint(x) \land 0 < x \to 1 \leq x\text{.} \]
            Recall that we adopted the convention that denominators of fractions are always positive, and the comparator circuit between two rationals $a/b$ and $c/d$ checks the integer inequality $ac < bd$. The consequent $1 \leq x$ stands for $1 = x \lor 1 < x$, which translates as $\operatorname{div}_b(a) = 1 \lor b < a$ when writing $x$ and $a/b$. Thus, the formula to prove is 
            \[ [a]_b = 0 \land 0 < a \to \operatorname{div}_b(a) = 1 \lor b<a\text{.} \]
            By L7.19, either $b< a$, $a < b$ or $a = b$. If $b < a$, we are done. If $a = b$, using L7.27, it is easy to show that $\operatorname{div}_a(a) = 1$, since by L7.29,
            \[ [a]_a = [0 + a]_a = [0]_a = 0 \]
            and thus
            \[  a = \operatorname{div}_a (a) \cdot a + [a]_a =  \operatorname{div}_a (a) \cdot a\text{.} \]
            Since by assumption $0 < a$, we have $a \neq 0$, so $\operatorname{div}_a(a) = 1$. Then,
            \[ \operatorname{div}_b(a) = \operatorname{div}_a(a) = 1\text{.} \]
            
            Finally, if $a < b$, we prove that the antecedent of the formula is falsified. We first show that $\operatorname{div}_a(b) = 0$. Suppose not, then it must be $\operatorname{div}_a(b) > 0$. 
            When taking $x = a$, $y = 0$, $p = b$ and $v = \operatorname{div}_b(a)$ in L7.28 above, we immediately get $b \leq a$, contradicting $a < b$.

            Now that we have $\operatorname{div}_b(a) = 0$, by L7.27 we get $a = [a]_b$. But if both $[a]_b = 0$ and $0 < a$, we get a contradiction. \qedhere
        \end{enumerate} 
\end{proof}

\begin{theorem}[Propositional translation for $\LAQ$]
\label{thm:new-propositional-translation}
    For every theorem $\varphi$ of $\LAQ$ and every object assignment $\sigma$, the propositional formula $||\varphi||_\sigma$ admits polynomial-size $\TC^0$-Frege proofs.
\end{theorem}

\begin{proof}
    The proof is analogous to Theorem 6.3 in \cite{cooksoltys}, except we need to handle the new axioms. By \Cref{lem:translation-new} above, the translations of the new axioms have short $\TC^0$-Frege proofs. This completes the proof.
\end{proof}

\section{Proof of \Cref{prop:equivalences}}
\label{app:equivalence}
We prove the equivalence between the machine-based and circuit-based definitions in the three settings.
\begin{enumerate}
        \item[(i)] \textbf{Classical automatability.} For the forward direction, suppose $A$ is an automating deterministic Turing machine. In order to simulate $A$ by a circuit, we need to introduce a uniform bound on the running time of $A$. We know $A$ runs in time $\mathsf{size}_S(\varphi)^c$ for some constant $c$. Consider now the machine $A'$ that takes as input both $\varphi$ and a size parameter $s$ in unary and runs $A(\varphi)$ for $s^c$ steps, and outputs a proof if one was found, and some other string otherwise. This machine $A'$ can be simulated by a uniform circuit family of size $O((|\varphi|+ s)^{2c})$, which is still polynomial in $|\varphi|+ s$, and which outputs a proof of size polynomial in $s$ if one exists.

        For the backwards direction, assuming a circuit family $\{ C_{n, s} \}_{n, s\in \bN}$, the machine on input $\varphi$ simulates $C_{|\varphi|, 1}(\varphi)$, $C_{|\varphi|, 2}(\varphi)$ and so on, checking every time whether the output proof is valid, up to the first value of $s$ for which a valid proof is obtained. This takes time polynomial in $\mathsf{size}_S(\varphi)$.
        
        \item[(ii)] \textbf{Randomized automatability.} The argument here is similar, except that we have to account for the equivalence between the bounded expected running time of the machine and the bounded error probability of the circuits.
        
        For the forward direction, let $R$ be a probabilistic machine automating $S$ in expected time $\mathsf{size}_S(\varphi)^c$ for some constant $c$. Let $T_\varphi$ be the random variable that denotes the number of steps $R$ takes to find a proof on input $\varphi$, when $\varphi$ does have some proof. We know that $\operatorname{E}[T_\varphi] \leq \size{S}{\varphi}^c$. Consider now the modified machine $R'$ that takes $\varphi$ and a size parameter $s$ and simulates $R(\varphi)$ for $k\cdot (n + s)^c$ steps, for some constant $k$ such as $k = 100$. This machine can be turned into a random circuit with $k\cdot (n + s)^c$ random bits. It just suffices to argue that for at least $2/3$ of the choices for the random bits, the circuit will output a proof when one exists. Indeed, by Markov's inequality, the probability that $R'$ might not output a proof in time $k\cdot (n + s)^c$ is just
        $\Pr[T_\varphi > k\cdot \operatorname{E}[T_\varphi]] \leq {1}/{k}$,
        which bounds the error of the circuit as desired.

        For the backwards direction, from the sequence $\{ C_{n, s}\}_{n, s \in \bN}$ of randomized circuits we get an error-bounded probabilistic Turing machine $R(\varphi, s)$ that first obtains the description of $C_{|\varphi|, s}$ (recall that the circuit family is uniform) and then simulates $C_{|\varphi|, s}(\varphi)$. This machine $R$ always halts after $(|\varphi| + s)^{O(1)}$ steps, and, whenever a proof of size $s^c$ exists, finds one with probability at least $2/3$. Now, consider the machine $R'$ that takes as input just the formula $\varphi$ and runs $R(\varphi, 1), R(\varphi, 2), \dots$ and so on, until a proof is found. For very small values of $s$ the the machine $R$ will never find a proof, because none exists. Once we get to values of $s$ large enough such that $s^c \geq \size{S}{\varphi}$, we might still be unlucky and not find a proof when running $R(\varphi, s)$, and move to $R(\varphi, s+1)$. Note, however, that the number of times we may increments the parameter $s$ before a proof is found follows a geometric distribution, and so the expected number of trials is at most $1/p$, where $p$ is the probability of success. Since $p$ is at least $2/3$, the expected number of times we will increment $s$ before a proof is found is at most $3/2$. Altogether, the machine $R'$ will run in expected time polynomial in $\size{S}{\varphi}$.
           
        \item[(iii)] \textbf{Quantum automatability.} The proof is identical to (ii). By Yao's result that quantum circuits can simulate quantum Turing machines running in time $T$ in size $O(T^2)$ \cite{yao}, we get the right transformations between circuits and machines, and the probability analysis is exactly the same.
    \end{enumerate} \qed

\section{Properties of random lattices (Proof of \Cref{lem:random-lattices})}\label{app:counting}
This section proves the two statements of \Cref{lem:random-lattices}. We start by proving that almost every randomly sampled matrix is full-rank (\Cref{lem:random-lattices}.i). We then prove two technical lemmas. Finally we show that almost every randomly sampled full rank matrix generates a lattice with no short vectors (\Cref{lem:random-lattices}.ii). 

From now on, unless otherwise specified, we consider lattices of the form $\LL_q(A)$ where $A \in \bZ_q^{m \times n}$ and $\rank(A) = n$. Note that for any such lattice $|\LL_q(A)|=q^n$.

\begin{proof}[Proof of \Cref{lem:random-lattices}.i]
When selecting a column vector there are $q^m$ different options. At each step $i$ the previously selected columns span a subspace of $\bZ_q^m$ with $q^{i-1}$ elements, meaning that on the $i$-th selection the odds of selecting a linearly dependent vector are only ${q^{i-1}}/{q^m} = {1}/{q^{m-i+1}}$. For each step $i$, this probability is less than ${1}/{q^{m-n+1}}$. By union bounding over the $n$ opportunities, the probability of ever selecting a linearly dependent column is less than the sum of these probabilities which in turn is less than ${n}/{q^{m-n+1}}$.
\end{proof}

\begin{lemma}\label{tec:lat-num}
Let $\{\LL_q\}$ be the set of all the possible distinct rank-$n$ lattices in $\bZ_q^m$. It holds that
\begin{equation*}
    |\{\LL_q\}| = \prod_{i=0}^{n-1}\left(\frac{q^m - q^i}{q^n - q^i}\right).
\end{equation*}
\end{lemma}
\begin{proof}
The cardinality of $\{\LL_q\}$ is equal to the number of rank-$n$ bases divided by the number of possible bases for each given lattice. Formally,
\[ |\{\LL_q\}| = \frac{|\{A| \rank(A) = n\}|}{|\{A' | \LL_q(A) = \LL_q(A')\}|}. \]
We take an algorithmic approach to counting the number of $A$ for which $\rank(A) = n$. To select such an $A$, we first set $a_0$ equal to one of the $q^m - 1$ non-zero points in $\mathbb{Z}_q^m$. Then for each subsequent $i$ we set $a_i$ equal to a point in $\mathbb{Z}_q^m$ not contained in the rank $i$ lattice spanned by $(a_0, \dots, a_{i-1})$. We know that there are $q^{i}$ vectors in that lattice leaving us with $q^m - q^{i}$ possible choices for $a_i$. To avoid double counting the permutations of a given basis we divide by $n!$, concluding that there are ${\prod_{i=0}^{n-1}\left(q^m - q^i\right)}/{n!}$ matrices $A$ with $\rank(A) = n$.

Next, to count the number of possible bases $A'$ we first note that any set of $n$ linearly independent vectors in $\LL_q(A)$ is a basis of $\LL_q(A)$. Then we follow the same method as above for generating a basis except our choices are now limited to the $q^n$ vectors in the lattice. So we end up with a total number of possible bases of ${\prod_{i=0}^{n-1}\left(q^n - q^i\right)}/{n!}$. 

If we divide the number of rank-$n$ bases by the number of bases per lattice we get $\prod_{i=0}^{n-1}\left(\frac{q^m - q^i}{q^n - q^i}\right)$.
\end{proof}

\begin{lemma}\label{tec:log-inequality} 
When $q \geq n \geq 1$, $\log_q(q+1)(n-1) \leq n$.
\end{lemma}
\begin{proof}
    Because $\log_q(q+1)$ is monotonically decreasing it suffices to show that $\log_n(n+1)(n-1) \leq n$. By change of basis and reordering this is equivalent to proving that
    \[ \frac{\ln(n+1)}{\ln(n)} \leq \frac{n}{n-1}. \]
    It is well known that \[ \frac{d(\ln(n))}{dn} = \frac{1}{n}, \] which is also monotonically decreasing, meaning that \[ \ln(n+1) \leq \ln(n) + 1\frac{1}{n}. \]
    Thus,
    \[ \frac{\ln(n+1)}{\ln(n)} \leq \frac{\ln(n) + \frac{1}{n}}{\ln(n)} = 1 + \frac{1}{n\ln(n)} \leq 1 + \frac{1}{n-1} = \frac{n}{n-1}\,. \]
\end{proof}

\begin{proof}[Proof of \Cref{lem:random-lattices}.ii]
    Every lattice with a short vector can be specified by a tuple of a lattice of rank $n-1$ and a short vector. All vectors with length less than $r$ must lie in the region $(-r,r)^m$ so we know there are fewer than $(2r+1)^m$ short vectors. Combining this with the number of lattices of rank $n-1$ from \Cref{tec:lat-num} we get that 
    \[ |\{\LL_q(A) \mid \lambda_1(\LL_q(A)) \leq r\}| \leq(2r+1)^m \prod_{i=0}^{n-2}\left(\frac{q^m - q^i}{q^{n-1} - q^i}\right). \]
    
    If we divide this upper bound on the number of lattices with short vectors by the exact count of the number of total lattices from  \Cref{tec:lat-num} we can see that the fraction of lattices which contain a vector of length less than $r$ is less than
    \begin{align*}
    \frac{(2r+1)^m \prod_{i=0}^{n-2}\left(\frac{q^m - q^i}{q^{n-1} - q^i}\right)}{\prod_{i=0}^{n-1}\left(\frac{q^m - q^i}{q^n - q^i}\right)} &= \frac{(2r+1)^m}{\frac{q^m-q^{n-1}}{q^n-q^{n-1}}} \prod_{i=0}^{n-2} \frac{\left(\frac{q^m - q^i}{q^{n-1} - q^i}\right)}{\left(\frac{q^m - q^i}{q^n - q^i}\right)} \\
    & = (2r+1)^m \frac{q^n - q^{n-1}}{q^m - q^{n-1}}\prod_{i=0}^{n-2}\frac{q^n-q^i}{q^{n-1}-q^i}\\
    & \leq (2r+1)^m \frac{q^n}{q^{m-1}}\left(\frac{q^n-q^{n-2}}{q^{n-1}-q^{n-2}}\right)^{n-1} \\
    &= \frac{(2r+1)^m}{q^{m-n-1}}(q+1)^{(n-1)} \\
    &= \frac{(2r+1)^m}{q^{m-n-1}}q^{\log_{q}(q+1)(n-1)}\\
&\leq \frac{(2r+1)^m}{q^{m-2n -1}}. \tag{by \Cref{tec:log-inequality}}
    \end{align*}
\end{proof}

\newpage
\section{Closed form for the basis of a $q$-ary lattices (Proof of \Cref{lem:q-basis})}\label{app:q-ary}

\Cref{lem:q-basis} follows immediately from the following technical statement, as the applications of $C$ and $C^{-1}$ simply change the indexing of the basis so that the first $n$ rows of the basis are linearly independent.

\begin{lemma}\label{$q$-ary lattice basis}
Let $\cB \in \bZ_q^{m \times n}$, with $\cB_1 \in \bZ_q^{n \times n}$ and $\cB_2 \in \bZ^{m - n \times n}$ so that
$$\cB = \begin{bmatrix}
\cB_1 \\
\cB_2
\end{bmatrix}.$$

If the the first $n$ rows of $\cB$ are linearly independent, then  $\Delta_q(\cB) = \LL(\cB')$ where $\cB' \in \bZ^{m \times m}$ is defined as 
$$\cB' = \begin{bmatrix}
I_n & 0\\
\cB_2\cB_1^-1 & qI_{m-n}
\end{bmatrix},$$
with $\cB_1^{-1}$ the inverse of $\cB_1$ in the modular field $\bZ^{m}$.
\end{lemma}


\begin{proof}
To prove equality we need to show $\LL(\cB') \subseteq \Delta_q(\cB)$ and $\Delta_q(\cB) \subseteq \LL(\cB')$. First we show  $\LL(\cB') \subseteq \Delta_q(\cB)$.

Assume $x =  \begin{bmatrix}x_1 \\ x_2 \end{bmatrix} \in \LL(\cB')$ where $x_1 \in \bZ^n$ and $x_2 \in \bZ^{m-n}$. Then we know there exists $z = \begin{bmatrix}z_1 \\ z_2 \end{bmatrix} \in \bZ^m$ such that $\cB' z = x$. Thus,

\begin{align*}
    x = \cB' z &= \begin{bmatrix}
I_n & 0\\
\cB_2\cB_1^-1 & qI_{m-n} 
\end{bmatrix}
\begin{bmatrix}
z_1\\
z_2
\end{bmatrix}
= \\
&= \begin{bmatrix}
z_1\\
\cB_2\cB_1^{-1}z_1 + q*z_2 
\end{bmatrix}
\equiv
\begin{bmatrix}
z_1\\
\cB_2\cB_1^{-1}z_1
\end{bmatrix}
\text{(mod q)} \\
&\equiv
\begin{bmatrix}
\cB_1\cB_1^{-1}z_1\\
\cB_2\cB_1^{-1}z_1
\end{bmatrix}
\text{(mod q)} \\
&=
\begin{bmatrix}
\cB_1\\
\cB_2
\end{bmatrix}
\begin{bmatrix}
\cB_1^{-1}z_1
\end{bmatrix}
\end{align*}


and so $x \in \Delta_q(\cB)$. Next we prove $\Delta_q(\cB) \subseteq \LL(\cB')$.

Assume $x = \begin{bmatrix}x_1 \\ x_2 \end{bmatrix}\in \Delta_q(\cB)$ then there is $z \in \bZ_q^m$ such that  

$$x \equiv \begin{bmatrix}
\cB_1\\
\cB_2
\end{bmatrix}
\begin{bmatrix}
z
\end{bmatrix}
(\text{mod } q)
= \begin{bmatrix}
\cB_1 z\\
\cB_2 z
\end{bmatrix}
= \begin{bmatrix}
x_1 + aq\\
x_2 + bq
\end{bmatrix},$$
meaning that $ z \equiv \cB_1^{-1}x_1$ and $\cB_2\cB_1^{-1}x_1 = x_2 + cq$. Now,

$$\cB' \begin{bmatrix}
x_1\\
-cq
\end{bmatrix}
=\begin{bmatrix}
x_1\\
x_2 + cq - cq
\end{bmatrix} = x,$$
so $\Delta_q(\cB) \subseteq \LL(\cB')$.
\end{proof}

\end{document}